\documentclass[a4paper,UKenglish,cleveref, autoref, thm-restate]{lipics-v2021}
\hideLIPIcs
\usepackage[utf8]{inputenc}

\usepackage{amsmath}
\usepackage{amsfonts}
\usepackage{amsthm}
\usepackage{amssymb}
\usepackage{mathtools}
\mathtoolsset{centercolon}
\usepackage{thmtools, thm-restate}
\usepackage{comment}

\usepackage{lmodern}
\usepackage{enumitem}
\usepackage{ifthen}
\usepackage{tikz}
\usepackage{url}
\usepackage{stmaryrd}
\usepackage{xparse}
\usepackage{bbm}
\usepackage{xcolor}
\usepackage{etoolbox}
\usepackage[subrefformat=parens]{subcaption}

\usepackage[most]{tcolorbox}

\usetikzlibrary{decorations.pathmorphing}
\usetikzlibrary{decorations.pathreplacing}
\usetikzlibrary{calc}
\usetikzlibrary{positioning}
\usetikzlibrary{patterns}
\usetikzlibrary{patterns.meta}

\pgfdeclarelayer{background}
\pgfdeclarelayer{foreground}
\pgfsetlayers{background,main,foreground}

\usepackage{todonotes}

\renewcommand{\phi}{\varphi}
\renewcommand{\epsilon}{\varepsilon}

\DeclareMathSymbol{\shortminus}{\mathbin}{AMSa}{"39}

\definecolor{colA}{RGB}{0,80,250}
\definecolor{colB}{RGB}{60,200,230}
\definecolor{colC}{RGB}{60,180,75}
\definecolor{colD}{RGB}{230,25,75}
\definecolor{colE}{RGB}{245,130,48}
\definecolor{colF}{RGB}{145,30,180}
\definecolor{colG}{RGB}{240,50,230}
\definecolor{colH}{RGB}{170,110,40}

\newcommand{\StructA}{\mathbf{A}}
\newcommand{\StructB}{\mathbf{B}}
\newcommand{\StructC}{\mathbf{C}}

\newcommand{\Ee}{\mathcal{E}}

\newcommand{\CSP}[1]{\mathrm{CSP}(#1)}

\newcommand{\restrict}[2]{#1|_{#2}}

\newcommand{\dom}[1]{\mathrm{dom}(#1)}

\newcommand{\bbN}{\mathbb{N}}
\newcommand{\bbZ}{\mathbb{Z}}

\newcommand{\Sym}[1]{S_{#1}}

\newcommand{\Qq}{\mathcal{Q}}
\newcommand{\Nn}{\mathcal{N}}
\newcommand{\Ll}{\mathcal{L}}
\newcommand{\Gg}{\mathcal{G}}
\newcommand{\Kk}{\mathcal{K}}
\newcommand{\Mm}{\mathcal{M}}
\newcommand{\Pp}{\mathcal{P}}

\newcommand{\FO}{\mathrm{FO}}
\newcommand{\BP}{\mathrm{BP}}

\renewcommand{\Sym}{\mathbf{Sym}}

\newcommand{\CSPtext}{\text{\upshape CSP}}
\newcommand{\NCSP}{\Nn_{\CSPtext}}

\DeclareMathOperator{\id}{id}

\newcommand{\Uu}{\mathcal{U}}

\DeclareMathOperator{\ar}{ar}

\nolinenumbers

\title{ Arity hierarchies for quantifiers closed under partial polymorphisms} 

\author{Anuj Dawar}{Department of Computer Science and Technology, University of Cambridge, UK}{anuj.dawar@cl.cam.ac.uk}{}{}

\author{Lauri Hella}{Faculty of Information Technology and Communication Sciences, Tampere University, Finland}{lauri.hella@tuni.fi}{}{}

\author{Benedikt Pago}{Department of Computer Science and Technology, University of Cambridge, UK}{benedikt.pago@cl.cam.ac.uk}{}{}

\keywords{finite model theory, constraint satisfaction problems, generalized quantifiers} 
\ccsdesc[500]{Theory of computation~Finite Model Theory}

\authorrunning{A. Dawar, L. Hella, B. Pago} 

\Copyright{Anuj Dawar, Lauri Hella, Benedikt Pago} 

\funding{Research funded in part by UK Research and Innovation (UKRI) under the UK government’s Horizon Europe funding guarantee: grant number EP/X028259/1.}

\acknowledgements{We thank Victor Lagerkvist for pointers to the literature on the use of partial polymorphisms in CSP algorithms. }

\begin{document}
	
	\maketitle
\begin{abstract}
We investigate the expressive power of generalized quantifiers closed under partial polymorphism conditions motivated by the study of constraint satisfaction problems.  We answer a number of questions 
arising from the work of Dawar and Hella (CSL 2024) where such quantifiers were introduced. 
For quantifiers closed under partial near-unanimity polymorphisms, we establish hierarchy results clarifying the interplay between the arity of the polymorphisms and of the quantifiers:
The expressive power of $(\ell+1)$-ary quantifiers closed under $\ell$-ary partial near-unanimity polymorphisms is strictly between the class of all quantifiers of arity $\ell-1$ and $\ell$.  We also establish an infinite hierarchy based on the arity of quantifiers with a fixed arity of partial near-unanimity polymorphisms.  Finally, we prove inexpressiveness results for quantifiers with a partial Maltsev polymorphism.  The separation results are proved using novel algebraic constructions in the style of Cai-F\"urer-Immerman and the quantifier pebble games of Dawar and Hella (2024).
\end{abstract}

\section{Introduction}
The study of generalized quantifiers, also known as Lindstr\"{o}m quantifiers \cite{ebbinghaus}, has played a major role in the development of finite model theory.  This theory is concerned with the expressive power of logics over finite structures.  When any particular class $\Kk$ of structures is not definable in a logic $L$, we can adjoin to $L$ a Lindstr\"{o}m quantifier $\Qq_\Kk$, which serves as an oracle for querying membership in $\Kk$. Extending $L$ with this quantifier yields, in a precise sense, the minimal extension of $L$ that can also express $\Kk$.  Thus, the investigation of the structure and interdefinability of quantifiers is the core of finite model theory.

Results about inexpressibility, commonly referred to as \emph{lower bound} results for logic, are often established by means of exhibiting winning strategies for model-comparison games.  Some of the most striking such results in finite model theory have come from model-comparison games developed for Lindstr\"{o}m quantifiers.  In particular, Hella's bijection game~\cite{Hella96} was introduced to establish the arity hierarchy for generalized quantifiers.  This game has since been widely used to establish many inexpressibility results in descriptive complexity and beyond. 
Hella's proof that there is a strict hierarchy of expressive power obtained by classifying quantifiers according to their arity has as a consequence that meaningful logics in descriptive complexity (such as those closed under first-order reductions) should include quantifiers of arbitrary arities. Including all quantifiers clearly leads to logics that are far too rich.  This leads us to study extensions of first-order logic with restricted families of quantifiers, but of unbounded arity.

A natural restriction is to require \emph{closure conditions} on the classes of structures $\Kk$ whose Lindström quantifiers $\Qq_\Kk$ we include in the logic. By default, such classes $\Kk$ are closed under isomorphism because logics never distinguish between isomorphic structures. Interestingly, by imposing stricter closure conditions, one can often obtain nicely behaved logics that admit a useful game characterization while still being rich in expressive power (i.e.\ stronger than first-order logic with counting quantifiers). One such example is first-order logic augmented by all linear-algebraic quantifiers~\cite{DGP19}. Its expressive power is characterized by the invertible map game, which was used by Lichter~\cite{Lichter23} in a highly sophisticated lower bound proof. The classes $\Kk$ defining linear-algebraic quantifiers are closed under an equivalence relation coarser than isomorphism, which is based on similarity of adjacency matrices over all finite fields \cite{makowskyFestschrift}.
More generally we can ask: \emph{For which closure conditions on classes $\Kk$ do the Lindström quantifiers $\Qq_\Kk$ give rise to logics that are both powerful and analyzable via a natural model-comparison game?}

The present paper should be seen as a contribution to this research programme. 
The closure condition we are concerned with here is closure under \emph{partial polymorphisms}, which was first studied by Dawar and Hella in 2024~\cite{dawarHella2024}. An $\ell$-ary polymorphism from a structure $\StructA$ to $\StructB$ is a homomorphism from the $\ell$-th power of $\StructA$ to $\StructB$. The interest in the notion comes from the theory of constraint satisfaction problems (CSP).
For a fixed finite structure $\StructB$, called the template, $\CSPtext(\StructB)$ denotes the class of structures $\StructA$ that map homomorphically to $\StructB$. The algebraic approach to CSP relates the computational complexity of $\CSPtext(\StructB)$ to the algebra of polymorphisms of $\StructB$ (that is, polymorphisms from $\StructB$ to itself) and more particularly to the equational theory of that algebra (see e.g.~\cite{barto}). The success of this programme culminated in the celebrated dichotomy theorem for finite-domain CSPs of Bulatov and Zhuk~\cite{Bulatov,Zhuk}.
Natural examples of quantifiers $\Qq_\Kk$ closed under partial polymorphisms are thus CSP quantifiers, that is, $\Kk = \CSPtext(\StructB)$ for a template structure $\StructB$ that also has the respective (partial) polymorphism. Such CSP quantifiers were first introduced by Hella~\cite{Hella23}, where a strict hierarchy of expressive power in terms of the size of $\StructB$ was demonstrated. This size hierarchy is quite different from the arity hierarchy that we establish here. Moreover, it should be stressed that families of quantifiers closed under partial polymorphisms are a more general concept and contain more than just CSP quantifiers.

Our main results are about the interaction of partial polymorphism conditions with restrictions on the arity of quantifiers. The conditions require that the class of structures defining a quantifier is closed under partial functions satisfying certain equational conditions. The conditions we are particularly interested in (already introduced in prior work~\cite{dawarHella2024}) are \emph{partial near-unanimity} polymorphisms and \emph{partial Maltsev} polymorphisms. Their total counterparts play a key role in the classification of tractable CSP, since templates with these polymorphisms give rise to CSP with relatively simple polynomial time algorithms \cite{barto2009constraint, BulatovDalmau2006}. 
The corresponding \emph{partial} polymorphisms, which we use here, are significant for another reason: For NP-complete CSP, they essentially govern the fine-grained exponential complexity under the strong exponential-time hypothesis, and they have applications in parameterized complexity as well (see e.g.\ \cite{JonssonLNZ17, LagerkvistW20, LagerkvistW22}).

Previous work~\cite{dawarHella2024} developed games for logics with quantifiers closed under partial polymorphisms and used these to establish that the problem of solving systems of equations over the Boolean field is not definable in the extension of infinitary logic with \emph{all} quantifiers closed under partial near-unanimity polymorphisms.

In the present paper we answer a number of natural questions arising from this previous work concerning arity hierarchies within the class of quantifiers closed under partial near-unanimity polymorphisms, and also establish the first expressivity bounds for quantifiers closed under partial Maltsev polymorphisms. To be specific, we establish in Section~\ref{sec:collapses} an arity reduction showing that any $r$-ary quantifier $Q$ which is closed under a partial near-unanimity polymorphism of arity $\ell \leq r$ can be defined using quantifiers of arity at most $\lceil \frac{\ell-1}{\ell} \cdot r  \rceil$. 
On the other hand, our main results in Section~\ref{sec:ArityHierarchy} show that the arity of such quantifiers does not collapse, and they exhibit in fact a strict infinite hierarchy of expressive power. First, we prove that quantifiers of arity $\ell+1$ with a partial near-unanimity polymorphism of arity $\ell$ are, in terms of expressive power, strictly between \emph{all} quantifiers of arity $\ell-1$ and all quantifiers of arity $\ell$ (Corollary \ref{cor:nearUnanimityStrictlySandwiched}). Secondly, we establish in Theorem~\ref{thm:hierarchyForFixedL} that for every fixed arity of polymorphisms, we can always gain in expressive power by increasing the arity of quantifiers. The results can be seen as a refinement of the arity hierarchy of Hella~\cite{Hella96}, and the separation results require novel constructions in the style of Cai-Fürer-Immerman \cite{caifurimm92} (called \emph{CFI-like} henceforth). In Section~\ref{sec:Maltsev} we move on to study a more powerful type of quantifiers, which cannot be fooled by CFI-like constructions, namely those with a partial Maltsev polymorphism.
Also for these, we obtain a strict arity hierarchy (Theorem~\ref{thm:maltsevHierarchy}). Moreover, we establish a first lower bound: When both the arity of the quantifiers as well as the number of variables are limited to $k$, Maltsev-closed quantifiers are strictly weaker than general ones
(Theorem~\ref{thm:MaltsevSeparation}). This proof requires playing the Maltsev-closed quantifier pebble game on another original algebraic construction.
It remains an intriguing open problem to prove lower bounds for first-order logic with Maltsev-closed quantifiers of arity $k$, but an unbounded number of variables. We believe that finding a suitable hard example would also inspire progress on other open problems, as it would have to be radically different from the CFI-like constructions typically used in many known lower bound results.

\section{Preliminaries}

We denote structures by boldface letters $\StructA,\StructB,\ldots$, and their universes by the corresponding Roman letters $A,B,\ldots$. All structures we consider in the paper are finite and relational. Thus, we assume without further mention that the universe $A$ of each structure $\StructA$ is finite and the vocabulary of $\StructA$ is a finite set of relation symbols.
Let $\StructA$ and $\StructB$ be $\tau$-structures for a vocabulary $\tau$. A function $f\colon C\to D$ is a \emph{partial isomorphism} from $\StructA$ to $\StructB$ if $C\subseteq A$, $D\subseteq B$ and $f$ is an isomorphism from the substructure of $\StructA$ induced by $C$ to the substructure of $\StructB$ induced by $D$.

Let $\StructA$ and $\StructB$ be $\tau$-structures. We write $\StructA\le\StructB$ if $A=B$ and $R^\StructA\subseteq R^\StructB$ for all $R\in\tau$. The \emph{union} $\StructA\cup\StructB$ of $\StructA$ and $\StructB$ is the $\tau$-structure with universe $A\cup B$ and relations $R^{\StructA\cup\StructB}\coloneqq R^\StructA\cup R^\StructB$ for each $R\in\tau$.

We denote first-order logic by $\FO$, and the infinitary logic with $k$ variables by $\Ll^k$. Furthermore, we write $\Ll\coloneqq \bigcup_{k\in\bbN}\Ll^k$.  We refer to \cite{EF} and \cite{Libkin} for precise definitions of these logics. Given a logic $L$ and two $\tau$-structures $\StructA$ and $\StructB$, we write $\StructA\equiv_L\StructB$ if $\StructA$ and $\StructB$ satisfy exactly the same $L$-sentences of vocabulary $\tau$.
For any positive integer $n$, we denote the set $\{1,\ldots,n\}$ by $[n]$.
\vspace{-1em}
\paragraph*{Generalized quantifiers}

Let $\tau=\{R_1,\ldots,R_m\}$ be a vocabulary and let $\Kk$ be a class of $\tau$-structures that is closed under isomorphisms. The extension $L(Q_\Kk)$ of a logic $L$ by the \emph{generalized quantifier} $Q_\Kk$ (also called \emph{Lindström quantifier}) is obtained by adding the following formula formation rule to the syntax of $L$:
\begin{itemize}
\item[] For any vocabulary $\sigma$, if $\Psi=(\psi_1,\ldots,\psi_m)$ is a tuple of $\sigma$-formulas, and $\bar{y}_1,\ldots,\bar{y}_m$ are tuples of variables such that $|\bar{y}_i|=\ar(R_i)$ for $i\in [m]$, then $Q_\Kk\bar{y}_1, \ldots,\bar{y}_m\Psi$ is a $\sigma$-formula.
A variable is free in this formula if, for some $i\in [m]$, it is free in $\psi_i$, but not in $\bar{y}_i$.
\end{itemize}
The sequence $\Psi=(\psi_1,\ldots,\psi_m)$ of $\sigma$-formulas defines an \emph{interpretation} of a $\tau$-structure in any $\sigma$-structure $\StructA$ with an assignment $\alpha$ that interprets all free variables of $Q_\Kk\bar{y}_1, \ldots,\bar{y}_m\Psi$: we define $\Psi(\StructA,\alpha)\coloneqq (A,\psi_1^{\StructA,\alpha},\ldots,\psi_m^{\StructA,\alpha})$, where $\psi_i^{\StructA,\alpha}\coloneqq \{\bar{b}\in A^{\ar(R_i)}\mid (\StructA,\alpha[\bar{b}/\bar{y}_i])\models\psi_i\}$ for $i\in [m]$.
The semantics of the formula $Q_\Kk\bar{y}_1, \ldots,\bar{y}_m\Psi$ can now be defined as follows:
\begin{itemize}
\item[] Let $\StructA$ be a $\sigma$-structure and $\alpha$ an assignment on $\StructA$ that interprets all free variables of $Q_\Kk\bar{y}_1, \ldots,\bar{y}_m\Psi$. Then $(\StructA,\alpha)\models Q_\Kk\bar{y}_1, \ldots,\bar{y}_m\Psi$ if, and only if, $\Psi(\StructA,\alpha)\in\Kk$.
\end{itemize}
The extension $L(\Qq)$ of a logic $L$ by a collection $\Qq$ of quantifiers is defined in the natural way by adding the corresponding formula formation and semantic rules for each quantifier $Q_\Kk$ in $\Qq$.
The \emph{arity} of a quantifier $Q_\Kk$ is the maximum arity of relations in the vocabulary of the class $\Kk$. For each $r\ge 1$, we denote the collection of all quantifiers of arity at most $r$ by $\Qq_r$.

\paragraph*{Bijective pebble games}

We apply the bijective pebble game of Hella~\cite{Hella96} in proving our arity hierarchy results in Section~\ref{sec:ArityHierarchy}.
Let $r\le k$ be positive integers, and  $\StructA$ and $\StructB$ be structures.

\begin{definition}
The \emph{$k,r$-bijection game} $\BP^k_r(\StructA,\StructB)$ is played between two players, \emph{Duplicator} and \emph{Spoiler}. Positions of the game are pairs $(\alpha,\beta)$, where $\alpha$ and $\beta$ are assignments on $\StructA$ and $\StructB$, respectively, such that $\dom\alpha=\dom\beta\subseteq X\coloneqq\{x_1,\ldots,x_k\}$. The initial position is $(\emptyset,\emptyset)$. 
Let $(\alpha,\beta)$ be the position after some round of the game. The next round consists of the following steps:
\begin{itemize}
\item[(1)] Duplicator chooses a bijection $f\colon A\to B$. If $|A|\not=|B|$, then Spoiler wins the play.
\item[(2)] Spoiler chooses an $r$-tuple $\bar{y}$ of distinct variables in $X$ and an $r$-tuple $\bar a\in A^r$.
\item[(3)] The next position is $(\alpha',\beta')$, where $\alpha'\coloneqq\alpha[\bar{a}/\bar{y}]$ and $\beta'\coloneqq\beta[f(\bar{a})/\bar{y}]$. If $\alpha'\mapsto\beta'$ is not a partial isomorphism $\StructA\to\StructB$, then Spoiler wins the play.
\end{itemize}
Duplicator wins the play if Spoiler does not win it in a finite number of rounds.
\end{definition}

The $k,r$-bijection game characterizes equivalence with respect to the infinitary $k$-variable logic with all $r$-ary quantifiers.

\begin{proposition}[\cite{Hella96}]
 Let $k \ge r$ be positive integers. Then Duplicator has a winning strategy in $\BP^k_r(\StructA,\StructB)$ if, and only if, $\StructA\equiv_{\Ll^k(\Qq_r)}\StructB$.
\end{proposition}

We say that a sentence $\psi$ of some logic separates two disjoint classes $\Kk$ and $\Kk'$ of structures if every structure in $\Kk$ satisfies $\psi$, but no structure in $\Kk'$ satisfies $\psi$.

\begin{corollary}
\label{cor:DuplicatorStrategyImpliesInseparatbilityOfClasses}
Fix a positive integer $r$.
Let $\Kk$ and $\Kk'$ be two disjoint classes of structures.
If for every $k \geq r$, there exist $\StructA \in \Kk$ and $\StructB \in \Kk'$ such that Duplicator has a winning strategy in $\BP^k_r(\StructA,\StructB)$, then no sentence in $\Ll(\Qq_r)$ separates $\Kk$ from $\Kk'$.
\end{corollary}

\paragraph*{Partial polymorphisms}
We go briefly through the basic notions related to families of partial polymorphisms. For a more detailed treatment and motivation, we refer to \cite{dawarHella2024}, where closure of generalized quantifiers with respect to such families was first considered. 

Let $p\colon A^\ell\to A$ be a partial function. Then for any $r\ge 1$, $p$ induces a partial function $\hat{p}\colon (A^r)^\ell\to A^r$ defined by applying $p$ ``column wise'': if $\bar{a}_i=(a_{i1},\ldots,a_{ir})\in A^r$ for $i\in [\ell]$, then $\hat{p}(\bar{a}_1,\ldots,\bar{a}_\ell)\coloneqq (p(a_{11},\ldots,a_{\ell 1}),\ldots,p(a_{1r},\ldots,a_{\ell r}))$. Naturally $\hat{p}(\bar{a}_1,\ldots,\bar{a}_\ell)$ is undefined if, and only if, $p(a_{1j},\ldots,a_{\ell j})$ is undefined for at least one $j\in [r]$. Sometimes we also view the tuples $\bar{a}_1, \dots, \bar{a}_\ell$ as the rows of an $\ell \times r$-matrix $M$ and write $\hat{p}(M)$ for $\hat{p}(\bar{a}_1,\ldots,\bar{a}_\ell)$.
Given a relation $R\subseteq A^r$, we write $p(R):=\{\hat{p}(\bar{a}_1,\ldots,\bar{a}_\ell)\mid \bar{a}_1,\ldots,\bar{a}_\ell\in R\}$. Furthermore, if $\StructA$ is a $\tau$-structure we denote by $p(\StructA)$ the $\tau$-structure with the same universe and relations $p(R^{\StructA})$, $R\in\tau$. 

We say that a partial function $p\colon A^\ell\to A$ is a \emph{partial polymorphism} of $\StructA$ if $p(\StructA)\le\StructA$.
A \emph{family $\Pp$ of $\ell$-ary partial functions} consists of a partial function $p_A\colon A^\ell\to A$ for every finite set $A$. The family \emph{respects bijections} if $f(p_A(a_1,\ldots,a_\ell))\simeq p_B(f(a_1),\ldots,f(a_\ell))$ whenever $f\colon A\to B$ is a bijection. Here $s\simeq t$ means that either $s=t$, or both $s$ and $t$ are undefined.

Let $\Pp$ be a family of partial functions that respects bijections, and let $Q_\Kk$ be a quantifier for a class $\Kk$ of $\tau$-structures. We say that $Q_\Kk$ is \emph{$\Pp$-closed} if for every $\StructB\in\Kk$ and every $\StructA\le\StructB\cup p_B(\StructB)$, we also have $\StructA\in\Kk$.
We denote the collection of all $\Pp$-closed quantifiers by $\Qq^\Pp$. 
It follows immediately from the definition that all $\Pp$-closed quantifiers $Q_\Kk$ are \emph{downwards monotone}: If $\StructB\in\Kk$ and $\StructA\le\StructB$, then $\StructA\in\Kk$.

The most interesting families of partial functions are based on equations on the components of the input tuple. Two examples of such families were studied in \cite{dawarHella2024}:

\begin{example}
~\\
\vspace{-1em}
\begin{enumerate}[label=(\alph*)]
\item The \emph{Maltsev} family $\Mm$ consists of functions $m_A\colon A^3\to A$ such that $m_A(a,a,b)=m_A(b,a,a)=b$ and $m_A(a,b,c)$ is undefined if $a\not=b$ and $b\not=c$.

\item The family $\Nn^\ell$ of $\ell$-ary \emph{partial near-unanimity functions} consists of functions $n^\ell_A\colon A^\ell\to A$ defined as follows: $n^\ell_A(a_1,\ldots,a_\ell)= a$ if at least $\ell-1$ of the inputs $a_i$ are equal to $a$; else, $n^\ell_A(a_1,\ldots,a_\ell)$ is undefined.
\end{enumerate}
Note that the definition of $n^\ell_A(a_1,\dots,a_\ell)$ is actually independent of the set $A$. Thus, we omit the subscript $A$ and write simply $n^\ell$ instead of $n^\ell_A$. 
\end{example}
It is straightforward to show that $n^\ell$ is a partial polymorphism of any structure with at most $\ell-1$-ary relations:

\begin{lemma}[{\cite[Lemma 9]{dawarHella2024}}]
\label{lem:TrivialPartPoly}
Let $\StructA$ be a $\tau$-structure such that $\ar(R)<\ell$ for all $R\in\tau$. Then $n^\ell$ is a partial polymorphism of $\StructA$.
\end{lemma}

\paragraph*{CSP quantifiers}

Let $\StructA$ and $\StructB$ be $\tau$-structures. A function $h\colon A\to B$ is a \emph{homomorphism} from $\StructA$ to $\StructB$ if for all $R\in\tau$ and $\bar{a}\in A^{\ar(R)}$, $\bar{a}\in R^\StructA$ implies that $h(\bar{a})\in R^\StructB$. We write $\StructA\to \StructB$ if there exists a homomorphism from $\StructA$ to $\StructB$.
For any fixed $\tau$-structure $\StructB$ we denote the class $\{\StructA\mid \StructA\to \StructB\}$ by $\CSP{\StructB}$. The class $\CSP{\StructB}$ is usually regarded as the (non-uniform) \emph{constraint satisfaction problem} (CSP): Given an input structure $\StructA$, decide whether there exists a homomorphism from $\StructA$ to $\StructB$. Here $\StructB$ is the \emph{template} of the CSP.
It is a consequence of the results of Bulatov and Zhuk~\cite{Bulatov,Zhuk} that, for any finite template $\StructB$, the corresponding CSP is either decidable in polynomial time or NP-complete, and which of these is the case depends only on the polymorphisms of $\StructB$.
A generalized quantifier $Q_\Kk$ is a \emph{CSP quantifier} if the defining class $\Kk$ is of the form $\CSP{\StructB}$ for some template $\StructB$. 
For CSP quantifiers, being closed under partial polymorphisms is a natural notion:
If $\StructB$ is a structure which has a (total) polymorphism that satisfies an identity, this yields a natural \emph{partial} polymorphism of the class of structures $\CSP\StructB$. Thus, quantifiers closed under this partial polymorphism are a generalization of CSP quantifiers based on a structure with the corresponding total polymorphism.
The following collections of CSP quantifiers are of particular interest to us.  
\begin{itemize}
\item For $r\ge 1$, $\Qq^{\text{CSP}}_r \subseteq \Qq_r$ is the collection of all CSP quantifiers of arity at most $r$. 
\item For $\ell\ge 3$, $\Qq^{\NCSP^\ell}\subseteq\Qq^{\Nn^\ell}$ is the collection of all $\Nn^\ell$-closed CSP quantifiers.
\item For $\ell\ge 3$ and $r\ge 1$, $\Qq^{\NCSP^\ell}_r\coloneqq\Qq^{\NCSP^\ell}\cap\Qq_r$.
\item $\Qq^{\Mm}$ is the collection of all quantifiers closed under the partial Maltsev polymorphism, and $\Qq^{\Mm}_r \coloneqq \Qq^{\Mm} \cap \Qq_r$, for every $r \geq 1$.
\end{itemize}	
Dawar and Hella \cite[Proposition 11]{dawarHella2024} proved that if $p_B$ is a partial polymorphism of $\StructB$, then  $Q_{\CSP{\StructB}}$ is $\Pp$-closed, assuming that $\Pp$ is projective\footnote{See \cite[Definition~6]{dawarHella2024} for the definition of projective.}. In particular, this holds for the family $\Nn_\ell$ for any $\ell\ge 3$. We conclude this section by proving a converse of this result. A relational structure is a \emph{core} if every homomorphism from it to itself is an automorphism. If a structure is not a core itself, it has a unique (up to isomorphism) substructure that is (see~\cite[Proposition~3.3]{Fagin2005}).
	
	\begin{lemma}
		\label{lem:closureOfClassImpliesTemplateClosure}
		Let $\Pp$ be a family of partial functions. Assume that $\StructB$ is a core such that the corresponding CSP quantifier $Q_{\CSP{\StructB}}$ is $\Pp$-closed. Then $p_B$ is a partial polymorphism of $\StructB$. 
	\end{lemma}	
	\begin{proof}
	We trivially have $\StructB \in \CSP{\StructB}$, and since $Q_{\CSP{\StructB}}$ is $\Pp$-closed, also $\StructB\cup p_B(\StructB) \in \CSP{\StructB}$. So there is a homomorphism $h\colon \StructB\cup p_B(\StructB) \to \StructB$. Now $h$ maps $R^\StructB$ to itself for every relation $R$ in the vocabulary $\tau$ of $\StructB$, and because $\StructB$ is a core, $h$ must be a permutation of $B$. But then we must have 
$p_B(\StructB) \le \StructB$ because otherwise, for some $R\in\tau$, $h$ would map a tuple in $R^{p_B(\StructB)} \setminus R^{\StructB}$ to a tuple in $R^{\StructB}$, which is not possible because $h$ induces a bijection between the tuples in $R^{\StructB}$. 
	\end{proof}

\begin{corollary}
\label{cor:quantifierClosedIffTemplateClosed}
Let $\StructB$ be a structure. The CSP quantifier $Q_{\CSP{\StructB}}$ is in $\Qq^{\NCSP^{\ell}}$ if, and only if, $n^\ell$ is a partial polymorphism of the core of $\StructB$.
\end{corollary}
\begin{proof}
	Note that $\CSP{\StructB}=\CSP{\StructB'}$ for the core $\StructB'$ of $\StructB$. Thus, the implication from left to right follows from Lemma~\ref{lem:closureOfClassImpliesTemplateClosure}, and the implication from right to left follows from \cite[Proposition 11]{dawarHella2024}.
\end{proof}

\section{Arity reductions}\label{sec:collapses}	

Our aim is to prove arity hierarchies for $\Nn^\ell$-closed quantifiers. Dawar and Hella \cite{dawarHella2024} introduced a model-comparison game for $\Ll^k(\Qq^{\Nn^\ell})$, and used it to prove that for any $\ell\ge 3$ there is an $\ell$-ary CSP quantifier in $\Qq^{\Nn^{\ell+1}}$ that is not definable in $\Ll(\Qq^{\Nn^\ell})$.
However, they left open the question whether for every fixed $\ell$, there is an arity hierarchy within the class of $\Nn^\ell$-closed quantifiers.
This question turns out to be nontrivial: In this section, we show that any $\Nn^\ell$-closed quantifier of arity at least $\ell$ can be defined by a quantifier with strictly smaller arity.  Moreover, the arity of some natural $\Nn^\ell$-closed quantifiers can even be reduced all the way down to $2$ in this way.

\paragraph*{Example: Hypergraph colouring}
For $m\ge r\ge 2$ let $H^r_m$ be the \emph{$r$-uniform hypergraph of size $m$}, i.e., $H^r_m:=([m],R^r_m)$, where $R^r_m$ is the set of all tuples $(a_1,\ldots,a_r)\in [m]^r$ such that $a_i\not= a_j$ for all $i\not=j$. Note that $H^2_m$ is just the complete graph $K_m$ on $m$ vertices. In \cite{dawarHella2024} the authors proved that the CSP quantifier corresponding to the class $\CSP{H^r_m}$ is $\Nn^\ell$-closed for any $\ell\ge 3$. 

Thus, one might think that the quantifiers $Q_{\CSP{H^r_m}}$, for $m\ge r\ge 2$, would form a natural arity hierarchy within $\Qq^{\Nn^\ell}$. However, this is not the case: Each $\CSP{H^r_m}$ can be defined already by the binary quantifier $Q_{\CSP{K_m}}$:
	
	\begin{restatable}{proposition}{hypergraphColouring}
		There is an $\FO$-definable reduction from $\CSP{H^r_m}$ to $\CSP{K_m}$.
	\end{restatable}	
    \noindent
    \textit{Proof sketch.} Given $\StructA$ in the vocabulary $\{R\}$ with $\ar(R)=r$, let $\Gg(\StructA)$ be the graph (possibly with self-loops) on vertex set $A$ with edges $\{(a_i,a_j)\mid (a_1,\ldots,a_r)\in R^\StructA, i\not=j\}$. Clearly $\Gg(\StructA)$ is $\FO$-definable in $\StructA$. It can be verified that $\StructA \to H^r_m$ if, and only if, $\Gg(\StructA)\to K_m$. \hfill \qedsymbol\\

\vspace{-2em}
	\paragraph*{Arity reductions for $\Nn^\ell$-closed quantifiers}	
    We now show a general arity reduction result for $N^\ell$-closed quantifiers. We start with the easier case of CSP quantifiers. 
	Let $r\ge 1$ and let $R \subseteq B^r$ be an $r$-ary relation. For $i,j \in [r]$, $i\le j$, let $p_{\setminus [i,j]}$ be the projection function $B^r\to B^{r-(j-i+1)}$ that maps $(b_1,\ldots,b_r)$ to $(b_1,\ldots,b_{i-1},b_{j+1},\ldots,b_r)$. Furthermore, let $R_{\setminus [i,j]}$ be the $r-(j-i+1)$-ary relation that is the projection of $R$ to $[r] \setminus \{i,i+1,...,j\}$, i.e., $R_{\setminus [i,j]} = \{ p_{\setminus [i,j]}(\bar{b}) \mid \bar{b}\in R  \}$.

	\begin{lemma}
	\label{lem:arityReductionTrick}
	Let $\ell \geq 3$ and $r \geq \ell$, and let $(i_1,j_1), (i_2,j_2),..., (i_{\ell},j_\ell)$ be a division of $[r]$ into $\ell$ intervals. Let $B$ be a finite set and $R\subseteq B^r$. Assume that $n^\ell$ is a partial polymorphism of $(B,R)$.
	Then $\bar{b} \in R$ if, and only if, $p_{\setminus [i_n,j_n]}(\bar{b})\in R_{\setminus [i_n,j_n]}$ for every $n \in [\ell]$.
\end{lemma}
\begin{proof} 
	Let $\bar{b}\in B^r$. If $\bar{b} \in R$, then $p_{\setminus [i_n,j_n]}(\bar{b})\in R_{\setminus [i_n,j_n]}$ by definition.
	For the other direction, assume that $p_{\setminus [i_n,j_n]}(\bar{b})\in R_{\setminus [i_n,j_n]}$ for every $n \in [\ell]$.
	This means that there exist tuples $\bar{c}_1,...,\bar{c}_\ell \in R$ such that $p_{\setminus [i_n,j_n]}(\bar{c}_n)=p_{\setminus [i_n,j_n]}(\bar{b})$ for each $n\in [\ell]$.
	Clearly then $\widehat{n^\ell}(\bar{c}_1,\ldots,\bar{c}_\ell)=\bar{b}$, and since $n^\ell$ is a partial polymorphism of $(B,R)$, it follows that $\bar{b}\in R$. 
\end{proof}

We use the following notation in the rest of this section. Let $\StructC$ be a $\tau$-structure, where $\ar(R)\le r$ for all $R\in\tau$, and let $(i_1,j_1), (i_2,j_2),..., (i_{\ell},j_\ell)$ be a division of $[r]$ into intervals. For each $R\in\tau$ and $n\in [\ell]$, let 
 \[
 	R^\StructC_n\coloneqq\begin{cases}R^\StructC_{\setminus [i_n,j_n]} & \text{if }j_n\le\ar(R), \\
 	R^\StructC_{\setminus [i_n,\ar(R)]} & \text{if } i_n\le\ar(R)<j_n, \\
	\emptyset & \text{if } \ar(R)<i_n.
	 \end{cases}
 \]
Then we denote by $\StructC^*$ the structure $(C,(R^\StructC_n)_{R\in\tau,n\in [\ell]})$. Note that $\StructC^*$ is $\FO$-definable in $\StructC$.

	\begin{lemma}
		\label{lem:CSPreductionByOne}
		Let $\ell,r$ and $(i_1,j_1), (i_2,j_2),..., (i_{\ell},j_\ell)$ be as in Lemma~\ref{lem:arityReductionTrick}, and let $\StructA$ and $\StructB$ be $\tau$-structures, where $\ar(R)\le r$ for all $R\in\tau$. Assume that $n^\ell$ is a partial polymorphism of $\StructB$. Then $\StructA \to \StructB$ if, and only if, $\StructA^* \to \StructB^*$.
	\end{lemma}	
	
\begin{proof}
		Let $h: \StructA \to \StructB$ be a homomorphism, and let $\bar{a}\in R^\StructA_n$ for some $R\in\tau$ and $n\in [\ell]$. Then there is a tuple $\bar{c} \in R^{\StructA}$ such that $\bar{a}=p_{\setminus [i_n,j_n]}(\bar{c})$. Since $h(\bar{c}) \in R^{\StructB}$, also $h(\bar{a})=p_{\setminus [i_n,j_n]}(h(\bar c)) \in R^{\StructB}_n$. Hence, $h$ is a homomorphism from $\StructA^*$ to $\StructB^*$. 
		
		Conversely, assume that $h$ is a homomorphism from $\StructA^*$ to $\StructB^*$. 
		Let $\bar{a} \in R^{\StructA}$ for some $R\in\tau$. Then for every $n\in [\ell]$, $p_{\setminus [i_n,j_n]}(h(\bar{a}))=h(p_{\setminus [i_n,j_n]}(\bar{a})) \in R^{\StructB}_n$. By Lemma \ref{lem:arityReductionTrick}, we have $h(\bar{a}) \in R^{\StructB}$.
	\end{proof}

	\begin{corollary}
		\label{cor:arityReductionQuantifiers}
		Let $r \ge \ell \ge 3$.
		Every quantifier in $\Qq_r^{\NCSP^{\ell}}$ is $\FO$-definable by a quantifier in $\Qq^{\CSPtext}_{\lceil \frac{\ell-1}{\ell} \cdot r \rceil}$.
	\end{corollary}	
	\begin{proof}
	Let $\StructB$ be  a $\tau$-structure such that $Q_{\CSP{\StructB}}$ is a quantifier in $\Qq_r^{\NCSP^{\ell}}$.  W.l.o.g. we can assume that $\StructB$ is a core, and hence  by Corollary \ref{cor:quantifierClosedIffTemplateClosed}, $n^\ell$ is a partial polymorphism of $\StructB$. 	
	Consider now a division $(i_1,j_1), (i_2,j_2),..., (i_{\ell},j_\ell)$ of $[r]$ into intervals such that the size of of each $[i_n,j_n]$ is at least $\lfloor \frac{r}{\ell}\rfloor$. Then the arity of each relation $R^\StructB_n$ of $\StructB^*$ is at most $r-\lfloor \frac{r}{\ell} \rfloor = \lceil \frac{\ell-1}{\ell} \cdot r \rceil$. Thus, $Q_{\CSP{\StructB^*}}\in\Qq^{\CSPtext}_{\lceil \frac{\ell-1}{\ell} \cdot r \rceil}$.
	
	By Lemma \ref{lem:CSPreductionByOne}, for any $\tau$-structure $\StructA$ we have $\StructA\in\CSP{\StructB}$ if, and only if, $\StructA^*\in\CSP{\StructB^*}$. Thus, the mapping $\StructA\mapsto\StructA^*$ is an $\FO$-definable reduction of $\CSP{\StructB}$ to $\CSP{\StructB^*}$.
\end{proof}	
	
	It should be noted that $n^\ell$ is usually not a partial polymorphism of the structure $\StructB^*$ in Lemma \ref{lem:CSPreductionByOne}. 
Otherwise we could iterate the arity reduction trick to obtain even lower arity quantifiers that define $\CSP{\StructB}$.

If we allow ourselves to use the infinitary logic $\Ll$ rather than just FO in the reductions, then the arity reduction works for \emph{any} $\Nn^\ell$-closed quantifier, not just CSPs. The general reduction consists in first closing the given structure $\StructA$ under $n^\ell$ (which is an $\Ll$-definable fixed-point induction), and then applying to it the FO-definable operator $^*$ defined before Lemma \ref{lem:CSPreductionByOne}. We prove in Appendix \ref{sec:appendixCollapseSection} that this yields the following arity reduction:
\begin{restatable}{lemma}{arityReductionNonCSP}
\label{lem:arityReductionNonCSP}
Let $r\ge\ell\ge 3$. Let $\Kk$ be an $\Nn^\ell$-closed class of $\tau$-structures, where $\ar(R)\le r$ for all $R\in\tau$. Then there exists a vocabulary $\sigma$ with at most $\lceil \frac{\ell-1}{\ell} \cdot r \rceil$-ary relations, a class $\Kk'$ of $\sigma$-structures, and an $\Ll$-definable reduction $f$ that maps $\tau$-structures to $\sigma$-structures in such a way that $\StructA \in \Kk$ if, and only if, $f(\StructA) \in \Kk'$. 
\end{restatable}

	\begin{corollary}
		\label{cor:arityReductionAllQuantifiers}
		Let $r\ge \ell \ge 3$.
		Every quantifier in $\Qq_r^{\Nn^{\ell}}$ is definable in $\Ll(\Qq_{\lceil \frac{\ell-1}{\ell} \cdot r \rceil})$.
	\end{corollary}	
Here and in Theorem \ref{thm:mainArityComparison} below, the expression $\lceil \frac{\ell-1}{\ell} \cdot r \rceil$ may equivalently be read as $r - i$, for the $i$ such that $i \cdot \ell \leq r < (i+1) \cdot \ell$.
Note that the arity reduction result in Corollary \ref{cor:arityReductionQuantifiers} for CSP quantifiers is not a special case of Corollary \ref{cor:arityReductionAllQuantifiers}, as the latter does not provide an $\FO$-definable reduction. Moreover, in Corollary \ref{cor:arityReductionQuantifiers}, the defining quantifier is again a CSP quantifier.

Summarizing our arity collapse results, we obtain the following picture.
First, we note that if the arity $\ell$ of the partial near-unanimity function $n^\ell$ is greater than the arity of the relations, then being $\Nn^\ell$-closed is not a restriction:
	
	\begin{theorem}\label{thm:ArityComparison}
	Let $\ell>r\ge 1$, $\ell\ge 3$. Then $\Ll(\Qq^{\Nn^\ell}_r) \equiv \Ll(\Qq_r)$ and $\Ll(\Qq^{\NCSP^\ell}_r) \equiv \Ll(\Qq^{\text{CSP}}_r)$.
	\end{theorem}
\begin{proof}	
The first claim follows from \cite[Lemma 9]{dawarHella2024} and the argument for \cite[Corollary 10]{dawarHella2024}.
The second claim follows from Corollary \ref{cor:quantifierClosedIffTemplateClosed} and Lemma~\ref{lem:TrivialPartPoly}.
\end{proof}
	
	Hence, only for $r \geq \ell$, the situation is interesting. In this case, we have:
	\begin{theorem}
	\label{thm:mainArityComparison}	
	Let $r \ge \ell \ge 3$. Then we have $\Ll(\Qq_{\ell-1}^{\CSPtext}) \leq \Ll(\Qq_r^{\NCSP^{\ell}}) \leq \Ll(\Qq^{\CSPtext}_{\lceil \frac{\ell-1}{\ell} \cdot r \rceil})$ and $\Ll(\Qq_{\ell-1}) \leq \Ll(\Qq_r^{\Nn^{\ell}}) \leq \Ll(\Qq_{\lceil \frac{\ell-1}{\ell} \cdot r \rceil})$.
	
	In particular, in the case $r = \ell$ we have $\Ll(\Qq_r^{\NCSP^{r}}) \equiv \Ll(\Qq^{\CSPtext}_{r-1})$ and $\Ll(\Qq_r^{\Nn^{r}}) \equiv \Ll(\Qq_{r-1})$.
	\end{theorem}
	\begin{proof}
	The inclusions $\Ll(\Qq_r^{\NCSP^{\ell}}) \leq \Ll(\Qq^{\CSPtext}_{\lceil \frac{\ell-1}{\ell} \cdot r \rceil})$ and $\Ll(\Qq_r^{\Nn^{\ell}}) \leq \Ll(\Qq_{\lceil \frac{\ell-1}{\ell} \cdot r \rceil})$ follow from Corollaries \ref{cor:arityReductionQuantifiers} and \ref{cor:arityReductionAllQuantifiers}. The inclusions $\Ll(\Qq_{\ell-1}^{\CSPtext}) \leq \Ll(\Qq_r^{\NCSP^{\ell}})$ and $\Ll(\Qq_{\ell-1}) \leq \Ll(\Qq_r^{\Nn^{\ell}})$ follow from Theorem \ref{thm:ArityComparison}, since clearly $\Ll(\Qq_{\ell-1}^{\NCSP^{\ell}}) \leq \Ll(\Qq_r^{\NCSP^{\ell}})$ and $\Ll(\Qq_{\ell-1}^{\Nn^{\ell}}) \leq \Ll(\Qq_r^{\Nn^{\ell}})$.
	\end{proof}

	\section{Arity hierarchies for partial near-unanimity quantifiers}\label{sec:ArityHierarchy}	
 Theorem \ref{thm:mainArityComparison} shows that restricting quantifiers of arity $r$ to those that are also $\Nn^r$-closed has the same effect on expressive power as reducing the arity to $r-1$, and does not really give rise to a new interesting logic.
    Thus, the natural question arises whether there are \emph{any} combinations of polymorphism arities $\ell$ and quantifier arities $r$ for which the $\Nn^\ell$-closed quantifiers yield a logic that is strictly different from any $\Ll(\Qq_{r'})$. In other words, we would like to know if the inclusions in Theorem \ref{thm:mainArityComparison} are strict if $r \neq \ell$. The following theorem establishes this already for the case $r = \ell+1$: 
	\begin{theorem}
		\label{thm:mainAritySeparation}	
		Fix $\ell \geq 3$. 
        Then $\Ll(\Qq_{\ell+1}^{\NCSP^{\ell}}) \not\leq \Ll(\Qq_{\ell-1})$. 
	\end{theorem}

    \begin{corollary}
    \label{cor:nearUnanimityStrictlySandwiched}
        For every $\ell \geq 3$,
    	$
		\Ll(\Qq_{\ell-1})\lneq\Ll(\Qq_{\ell+1}^{\Nn^{\ell}}) \lneq \Ll(\Qq_{\ell}).
		$
    \end{corollary}
    \begin{proof}
    The first inclusion follows with Theorem \ref{thm:mainArityComparison}, and the fact that it is strict is a consequence of Theorem \ref{thm:mainAritySeparation}, which shows that already the fragment $\Ll(\Qq_{\ell+1}^{\NCSP^{\ell}})$ of $\Ll(\Qq_{\ell+1}^{\Nn^{\ell}})$ is not contained in $\Ll(\Qq_{\ell-1})$.
    The second inclusion is a consequence of Theorem \ref{thm:mainArityComparison}, and $\Ll(\Qq_{\ell}) \not\leq \Ll(\Qq_{\ell+1}^{\Nn^{\ell}})$ was shown in \cite[Theorem 30]{dawarHella2024}.
    \end{proof}
    A second natural question that we answer in this section is whether the polymorphism-closed quantifiers exhibit an arity hierarchy with respect to the arity of the quantifiers (keeping the arity of the polymorphism fixed): For the class of all generalized quantifiers, Hella~\cite{Hella96} proved that $\Ll(\Qq_r) \lneq \Ll(\Qq_{r+1})$ for all $r \in \bbN$. The following result shows that for every fixed polymorphism-arity $\ell$, one can always gain strictly more expressive power by sufficiently increasing the arity of the quantifiers.   
	\begin{theorem}
		\label{thm:hierarchyForFixedL}		
		Let $r\ge\ell \ge 3$. Then $\Ll(\Qq^{\NCSP^{\ell}}_{q(r,\ell)}) \nleq \Ll(\Qq_{r})$,
        where $q(r,\ell) = (2r+1)(\lfloor\frac{2r+1}{\ell-1}\rfloor+2)$.
	\end{theorem}

    \begin{corollary}
    Let $r\ge\ell \ge 3$. Then there is an $r' > r$ such that $\Ll(\Qq^{\NCSP^{\ell}}_{r}) \lneq \Ll(\Qq^{\NCSP^{\ell}}_{r'})$.
    \end{corollary}
    Note that the separation $\Ll(\Qq_{\ell+1}^{\NCSP^{\ell}}) \not\leq \Ll(\Qq_{\ell-1})$ in Theorem \ref{thm:mainAritySeparation} is tighter than the one that is implied by Theorem \ref{thm:hierarchyForFixedL} for $r=\ell+1$. So neither of the two results implies the other.

    The proofs of Theorems \ref{thm:mainAritySeparation} and \ref{thm:hierarchyForFixedL} are structured similarly but employ different constructions. In both cases we construct a CSP template structure whose associated CSP quantifier can be used to distinguish certain CFI-like structures, which we view as the CSP instances. These instances are such that they cannot be distinguished in the logic $\Ll(\Qq_r)$, for the respective value of $r$ that we have to beat in the theorems. The novelty in both constructions is that we enforce the template structure to be closed under the $\ell$-ary partial near-unanimity function, so that the corresponding CSP quantifier is indeed in $\Qq^{\NCSP^{\ell}}$. 

    The rest of the section is divided into three parts: First, construction of the respective templates needed for the proofs of the two theorems, then construction of the instances, and finally, we provide winning strategies for Duplicator in the bijective pebble game to prove indistinguishability of the instances in the respective logic $\Ll(\Qq_r)$.

    \subsection{Construction of polymorphism-closed templates}

    The templates we construct have as relations the solution spaces of certain systems of linear equations. To make them closed under the partial near-unanimity function $n^\ell$, we  ensure that no collection of tuples in a relation satisfies the conditions for applicability of $n^\ell$, except in cases where $n^\ell$ outputs a tuple that is already in the collection. Then $n^\ell$ is a partial polymorphism of the template.

    \paragraph*{Template structure for the proof Theorem \ref{thm:mainAritySeparation}}
    We start by developing the ideas to forbid the applicability of $n^\ell$.
    We say that an $n \times m$-matrix over a ring or field has the \emph{$n$-near-unanimity} property if in every column, all entries except possibly one (the \emph{minority entry}), are equal. 
    Throughout this section, arithmetic operations are always meant in the ring $\bbZ_3$. Omitted/shortened proofs can be found in Appendix \ref{sec:appendixArityHierarchySection}.

	\begin{lemma}
		\label{lem:baseCaseSeparatingVector}
		Let $A \in \bbZ_3^{3 \times 3}$ be a matrix satisfying:
		\begin{enumerate}
			\item $A$ has the $3$-near-unanimity property.
			\item $\sum_{j \in [3]} a_{1j} = \sum_{j \in [3]} a_{2j} = \sum_{j \in [3]} a_{3j}$, that is, the row-sums are all equal.
			\item $\widehat{n^3}(A)$ is a triple that is not a row of $A$.
		\end{enumerate}	
		Then $A \cdot (0,1,2)^T$ is a vector in $\bbZ_3^3$ with all three entries distinct.
	\end{lemma}	
	\begin{proof}
		Assumption 3) entails that no two rows of $A$ are identical. 
		Since no two rows of $A$ are equal, the row $A_{2-}$ differs from $A_{1-}$ in at least one position. They cannot differ in all positions because then, by the first property, $A_{3-}$ must be identical to either $A_{1-}$ or $A_{2-}$, which cannot be the case. If $A_{1-}$ and $A_{2-}$ differ in only one position, then $\widehat{n^3}(A) \in \{A_{1-}, A_{2-} \}$, so this also contradicts assumption 3). Thus, $A_{1-}$ and $A_{2-}$ differ in exactly two positions. Assume first that $a_{11} \neq a_{21}$ and $a_{12} \neq a_{22}$, but $a_{13} = a_{23}$. 
		Let $v_1 = A_{1-} \cdot (0,1,2)^T$ and $v_2 = A_{2-} \cdot (0,1,2)^T$. Then $v_2 - v_1 = a_{22} - a_{12} \neq 0$, so $v_1 \neq v_2$. The same argument works in case that $a_{11} \neq a_{21}$ and $a_{13} \neq a_{23}$. In case that $a_{11} = a_{21}$, $a_{12} \neq a_{22}$, and $a_{13} \neq a_{23}$, we have $v_2 - v_1 = a_{22} - a_{12} + 2(a_{23} - a_{13})$. By condition 2) and because $a_{11} = a_{21}$, we have $a_{12}+a_{13} = a_{22}+a_{23}$. So 
		\[
				v_2 - v_1 = a_{22} - a_{12} + 2(a_{23} - a_{13}) = a_{22}+a_{23}+a_{23} - (a_{12}+a_{13}+a_{13}) = a_{23}-a_{13} \neq 0.
		\]
		Therefore, in all cases, $v_1 \neq v_2$. The same reasoning can be applied to any pair of rows, and thus also shows that $v_3 = A_{3-} \cdot (0,1,2)^T \neq v_1$, and $v_3 \neq v_2$.
	\end{proof}	

The lemma also generalizes to square matrices of larger dimensions because one can always find a $(3 \times 3)$-submatrix with the properties required by Lemma \ref{lem:baseCaseSeparatingVector} in them.

	\begin{restatable}{lemma}{separatingVector}
		\label{lem:separatingVector}
		Let $\ell \geq 3$ and
		let $A \in \bbZ_3^{\ell \times \ell}$ be a matrix satisfying:
		\begin{enumerate}
			\item $A$ has the $\ell$-near-unanimity property.
			\item The row-sums of $A$ are all equal.
			\item $\widehat{n^\ell}(A)$ is not a row of $A$.
		\end{enumerate}	
		Then $A \cdot (0,1,2,0,0,...,0)^T$ is a vector in $\bbZ_3^\ell$ that has at least three distinct entries.
	\end{restatable}	

We also need to consider the case where $\widehat{n^\ell}(A)$ is in fact a row of $A$. The following can be shown with a combinatorial argument. 
\begin{restatable}{lemma}{conditionForGettingARowBack}
	\label{lem:conditionForGettingARowBack}
	Let $\ell \geq 3$, and let $A \in \bbZ_3^{\ell \times \ell}$ be a matrix that satisfies the first two conditions from Lemma \ref{lem:separatingVector}.
	Then $\widehat{n^\ell}(A)$ is a row of $A$ if and only if two rows of $A$ are identical. 
\end{restatable}

For each $\ell \geq 3$ as in Theorem \ref{thm:mainAritySeparation}, we now define the template structure $\StructB_\ell$ that is used to show the separation $\Ll(\Qq_{\ell+1}^{\NCSP^{\ell}}) \not\leq \Ll(\Qq_{\ell-1})$. 

    \begin{definition}
    \label{def:template1}
    Let $\ell \geq 3$ be fixed.
    Let $\bar{u}$ denote the vector $(0,1,2,0,\dots,0) \in \bbZ_3^\ell$. 
    Let $\StructB_\ell = (\bbZ_3, R_0, R_1, R_2)$ with each $R_j$ defined as follows. 
	For each $j \in \{0,1,2\}$, let
	$
	R_j^{\StructB_\ell} \coloneqq \{  (a_1, ..., a_\ell, \bar{u} \cdot \bar{a}) \mid \bar{a} \in \bbZ_3^{\ell}, \sum_{i \in [\ell]} a_i = j   \}.
	$
    \end{definition}
	\begin{lemma}
	\label{lem:modifiedTemplatePolyClosed}	
	For each $\ell \geq 3$, $n^\ell$ is a partial polymorphism of the  structure $\StructB_\ell$.
	\end{lemma}
	\begin{proof}
		Let $j \in \bbZ_3$ and $\ell \geq 3$. 
		Let $\bar{b}_1,\dots, \bar{b}_\ell \in R_j^{\StructB_{\ell}}$. We show that either, $\widehat{n^\ell}(\bar{b}_1,\dots, \bar{b}_\ell) \in \{\bar{b}_1,\dots, \bar{b}_\ell\}$, or $\widehat{n^\ell}(\bar{b}_1,\dots, \bar{b}_\ell)$ is undefined.
		Let $A \in \bbZ_3^{\ell \times \ell}$ be the matrix whose rows are the projections of the tuples $\bar{b}_1,\dots, \bar{b}_\ell$ to their first $\ell$ components.
		If $A$ does not have the $\ell$-near-unanimity property, then $\widehat{n^\ell}(A)$ is undefined, and hence $\widehat{n^\ell}(\bar{b}_1,\dots, \bar{b}_\ell)$ is undefined. 
		So suppose $A$ has the $\ell$-near-unanimity property. The row-sums of $A$ are all equal to $j$ by definition of $R_j^{\StructB_{\ell}}$. 
		If $\widehat{n^\ell}(A)$ is not a row of $A$, then $A$ satisfies the assumptions of Lemma \ref{lem:separatingVector}. Hence $A \cdot (0,1,2,0,0,\dots,0)^T$ contains three distinct entries. This means that $n^\ell(\bar{b}_{1(\ell+1)},\dots, \bar{b}_{\ell(\ell+1)})$ is undefined. If $\widehat{n^\ell}(A)$ is a row $\bar{a}$ of $A$, then by Lemma \ref{lem:conditionForGettingARowBack}, there are two identical rows in $A$, and these must be equal to $\bar{a}$. Therefore, the value $\bar{a} \cdot \bar{u}$ also appears at least twice as the $(\ell+1)$st component of the tuples $\bar{b}_1,\dots, \bar{b}_\ell$.
        Hence, if $n^\ell(\bar{b}_{1(\ell+1)},\dots, \bar{b}_{\ell(\ell+1)})$ is defined, it yields the correct value $\bar{a} \cdot \bar{u}$, and thus $\widehat{n^\ell}(\bar{b}_1,\dots, \bar{b}_\ell) \in \{\bar{b}_1,\dots, \bar{b}_\ell\}$.
	\end{proof}

    \paragraph*{Template structure for the proof Theorem \ref{thm:hierarchyForFixedL}}

    In Theorem \ref{thm:hierarchyForFixedL}, the arity of the quantifiers can be much larger than the arity of the polymorphisms. Thus, the matrices of interest are no longer square, but rectangular, and we cannot directly use Lemma \ref{lem:separatingVector}. Instead, we use the following. It can be shown with similar reasoning as in the proof of Lemma \ref{lem:baseCaseSeparatingVector}.

	\begin{restatable}{lemma}{separatingVectorPairs}
		\label{lem:separatingVectorPairs}
		Let $\ell \geq 3$, $r \geq \ell$, $q=\lfloor \frac{r}{\ell-1}\rfloor$, and $p=r-q(\ell-1)$. 
		There is a collection $\Uu \subseteq \bbZ_3^{r}$ of size $|\Uu| = (\ell-1)q(q-1)+2pq$ such that 
        for every matrix $A \in \bbZ_3^{\ell \times r}$ with the $\ell$-ary near-unanimity property, one of the following two is the case:
		\begin{enumerate}
			\item If $\widehat{n^\ell}(A)$ is not a row of $A$, there is a vector $\bar{u} \in \Uu$ such that $A \cdot \bar{u}$ contains three distinct entries. 
            \item If $\widehat{n^\ell}(A)$ is a row $\bar{a}$ of $A$, then either there exists a $\bar{u} \in \Uu$ such that $A \cdot \bar{u}$ contains three distinct entries, or for all $\bar{u} \in \Uu$, $\widehat{n^\ell}(A \cdot \bar{u}) = \bar{a} \cdot \bar{u}$.
		\end{enumerate}	
	\end{restatable}
    \noindent
	\textit{Proof sketch.}
		Split $[r]$ into sets $K_1 \uplus\cdots \uplus K_{\ell-1} = [r]$ such that $|K_1|=\cdots=|K_p|=q+1$ and $|K_{p+1}|=\cdots=|K_{\ell-1}|=q$. 
		Define $\Uu$ as follows. For every $2$-element set of columns $\{j_1,j_2\}$ such that $\{j_1,j_2\} \subseteq K_i$ for some $i\in [\ell-1]$, include in $\Uu$ two vectors  $\bar{u}^{11}, \bar{u}^{12} \in \bbZ_3^{r}$ such that $\bar{u}^{11}_{j_1} = 1, \bar{u}^{11}_{j_2} = 1$, and $\bar{u}^{12}_{j_1} = 1, \bar{u}^{12}_{j_2} = 2$. In all other components, the two vectors are zero.
		Note that $|\Uu| = 2 \bigl(p \binom{q+1}{2}+(\ell-1-p)\binom{q}{2}\bigr) = 2 \bigl( p\frac{(q+1)q}{2}+(\ell-1-p)\frac{q(q-1)}{2}\bigr)=(\ell-1)q(q-1)+2pq$.
		
		Assume first that $\widehat{n^\ell}(A)$ is not a row of $A$.
        Then $A$ must have $\ell$ non-constant columns whose minority entry is in a different row, respectively. So there is an $i \in [\ell-1]$ such that $K_i$ contains at least two of these columns, call them $j_1, j_2$. By definition of $\Uu$, the vectors $\bar{u}^{11}, \bar{u}^{12}$ for $j_1, j_2$ are in $\Uu$. Then it can be verified that $A \cdot \bar{u}^{11}$ or $A \cdot \bar{u}^{12}$ contains three distinct entries.
        
        Now assume that $\widehat{n^\ell}(A) = \bar{a}$ is a row of $A$. It can be checked that if $\widehat{n^\ell}(A \cdot \bar{u})$ is defined for all $\bar{u} \in \Uu$, then for every $\bar{u} \in \Uu$, there must be at least two rows of $A \cdot \bar{u}$ whose entry is $\bar{a} \cdot \bar{u}$. Then it follows that $\widehat{n^\ell}(A \cdot \bar{u}) = \bar{a} \cdot \bar{u}$ for all $\bar{u} \in \Uu$, as desired. \hfill \qedsymbol\\

    The template structure used for Theorem \ref{thm:hierarchyForFixedL} is similar to the previous one, except that it has more extra entries in the relations to forbid the applicability of $n^\ell$ (using the set of vectors $\Uu$ from the above lemma). 
    For $r \geq \ell$, the instances of this template are supposed to be indistinguishable in $\Ll(\Qq_{r})$. But the construction using the vectors in $\Uu$ effectively makes it possible to fix two entries of the original relation with just one variable. To ensure $\Ll(\Qq_{r})$-indistinguishability of the instances, we therefore need to base the template on $(2r+1)$-ary relations, augmented by the products with all vectors in $\Uu$.
    For every $\ell \geq 3$ and $r \geq \ell$, we define the template structure $\StructB^*_{r, \ell}$ over $\bbZ_3$ as follows. 
    \begin{definition}
    \label{def:template2}
    Let $\ell \geq 3$ and $r \geq \ell$.
    Let $\Uu$ be the set of vectors from Lemma \ref{lem:separatingVectorPairs}, for $(2r+1)$ instead of $r$. 
	Note that then $m \coloneqq |\Uu| =(\ell-1)q(q-1)+2pq$, where $q=\lfloor\frac{2r+1}{\ell-1}\rfloor$ and $p=2r+1-\lfloor\frac{2r+1}{\ell-1}\rfloor(\ell-1)<\ell-1$. Thus, $m<(\ell-1)q(q-1)+2(\ell-1)q=(\ell-1)q(q+1)\le (2r+1)(\lfloor\frac{2r+1}{\ell-1}\rfloor+1)$.
	In $\StructB^*_{r, \ell}$, we define the relations $R_j$, for $j \in \bbZ_3$ as:
	$
	R_j^{\StructB^*_{r,\ell}} \coloneqq \{ (a_1,\dots,a_{2r+1}, \bar{a} \cdot \bar{u}_1,\dots,\bar{a} \cdot \bar{u}_m ) \mid \sum_{i\in [2r+1]} a_i = j  \},
	$
	where $\bar{u}_1,...,\bar{u}_m$ is an enumeration of the set $\Uu$. 
    \end{definition}
    Note that for fixed $\ell$, the arity of the template structure is quadratic in $r$.

    \begin{lemma}
	\label{lem:modifiedTemplate2PolyClosed}	
	For all $\ell \geq 3$ and $r \geq \ell$, $n^\ell$ is a partial polymorphism of the  structure $\StructB^*_{r,\ell}$.
	\end{lemma}
    \begin{proof}
    Analogous to the proof of Lemma \ref{lem:modifiedTemplatePolyClosed}, now using Lemma \ref{lem:separatingVectorPairs}.
    \end{proof}

	\subsection{Constructing instances}

	We aim to prove $\Ll(\Qq_{\ell+1}^{\NCSP^{\ell}}) \nleq \Ll(\Qq_{\ell-1})$ 
    by exhibiting two classes $\Kk_\ell, \tilde{\Kk}_\ell$ of structures that cannot be distinguished by any sentence in $\Ll(\Qq_{\ell-1})$ but are separated by a sentence in $\Ll(\Qq_{\ell+1}^{\NCSP^{\ell}})$.
    Likewise, we prove $\Ll(\Qq^{\NCSP^{\ell}}_{q(r,\ell)}) \nleq \Ll(\Qq_{r})$ by constructing classes $\Kk^*_{r,\ell}, \tilde{\Kk}^*_{r,\ell}$ of structures that cannot be distinguished by any sentence in $\Ll(\Qq_{r})$ but are separated by a sentence in $\Ll(\Qq^{\NCSP^{\ell}}_{q(r,\ell)})$.
    Both classes of structures are obtained with a CFI-like construction over an $\ell$-regular (or $(2r+1)$-regular, respectively) base graph.  
    With each vertex $v$ in the base graph, we associate a substructure, also called \emph{vertex gadget}.
    In the first construction, this gadget is denoted $\StructA(v,s)$. In the second construction it is $\StructA^*(v,s)$. In both cases, $s \in \bbZ_3$ is a ``charge'' we associate with $v$.
	We only present $\StructA(v,s)$ here. The gadget $\StructA^*(v,s)$ is analogous, but such that it matches the template relations from Definition \ref{def:template2} instead of Definition \ref{def:template1}.
    For $v \in V(G)$, let $E(v) \subseteq E(G)$ denote the set of edges incident to $v$.
	\begin{definition}[Vertex gadgets]
    \label{def:Components}
    ~\\
    Fix $\ell \geq 3$ and let $G = (V,E)$ be a connected $\ell$-regular graph with an order $\leq^G$ on its vertices.
	Let $\bar{u} \in \bbZ_3^\ell$ be the vector from Definition \ref{def:template1}, $v\in V$ and $\bar{e}(v) = (e_1, \dots , e_{\ell})$ the set $E(v)$, written as a tuple in the order induced by $\leq^G$. Then for each $s\in\bbZ_3$, we let $\StructA(v,s)\coloneqq(A_v,R_0^{\StructA(v,s)},R_1^{\StructA(v,s)},R_2^{\StructA(v,s)})$, where
	\begin{itemize}
	\item $A_v\coloneqq B_v\cup C_v$ for $B_v\coloneqq E(v)\times\bbZ_3$ and $C_v\coloneqq \{v\} \times\bbZ_3$,
	\item $R_j^{\StructA(v,s)}\coloneqq \{((e_1,a_1),\ldots,(e_{\ell},a_{\ell}),(v,\bar{a}\cdot\bar{u})) \mid \sum_{i\in[\ell]}a_i=j-s\}$ for each  $j\in\{0,1,2\}$,
	\end{itemize}
    \end{definition}

    We can now define our CFI structures that separate the desired logics. Starting with a suitable base graph, we simply replace every vertex $v$ with a gadget of the form $\StructA(v,s)$ or $\StructA^*(v,s)$.
    We present in the following only the construction that uses the gadgets $\StructA(v,s)$ (for Theorem \ref{thm:mainAritySeparation}). The other construction for the proof of Theorem~\ref{thm:hierarchyForFixedL} and its relevant properties are completely analogous.
    Putting together the structures $\StructA(v,s)$ we have defined for each vertex of $G$, we obtain our CFI structures:
	\begin{definition}\label{def:CFIstruct}
	Let $G = (V,E,\leq^G)$ be an $\ell$-regular connected graph with a linear order on its vertices. 
	\begin{enumerate}[label=(\alph*)]
	\item For each $U\subseteq V$, we define 
$\StructA(G,U)= (A_G,R_0^{\StructA(G,U)},R_1^{\StructA(G,U)},R_2^{\StructA(G,U)})$, where
	\begin{itemize}
	\item $A_G\coloneqq \bigcup_{v\in V}A_v$,
	\item $R_j^{\StructA(G,U)}\coloneqq \bigcup_{v\in V\setminus U}R_j^{\StructA(v,0)}\cup \bigcup_{v\in U}R_j^{\StructA(v,1)}$ for each $j\in\{0,1,2\}$.
	\end{itemize}
	
	\item Furthermore, we define $\StructA(G)\coloneqq\StructA(G,\emptyset)$ and 
	$\tilde{\StructA}(G)\coloneqq\StructA(G,\{\tilde{v}\})$, where $\tilde{v}\in V$ is the smallest vertex of $G$ with respect to the order $\leq^G$.
    \end{enumerate}
	\end{definition}
    We write $\StructA^*(G)$ and $\tilde{\StructA}^*(G)$ to refer to the analogous structures obtained by using the vertex gadgets $\StructA^*(v,s)$ instead of $\StructA(v,s)$.

     To define the classes of structures that separate the logics, we apply the above construction to base graphs similar to the ones used by Hella in \cite{Hella96}. For a number $k \in \bbN$, we call a graph \emph{$k$-connected} if it is connected and remains so after the removal of any $k_1$ vertices and $k_2$ edges, for any $k_1 + k_2 < k$. Arbitrarily large graphs having this and other useful properties exist, for example constructions based on higher-dimensional grids.
     \begin{restatable}{lemma}{existenceOfGridGraphs}
        \label{lem:existenceOfGridGraphs}
        For every $d \geq 2$ and for infinitely many $n \in \bbN$, there exists a $d$-regular bipartite graph that is $d$-connected and satisfies that each part of the bipartition contains exactly $n$ vertices.
     \end{restatable}
    For fixed $d,n \in \bbN$, we now define a graph that we denote $G_{d,n}$. This will be used as the base graph for our construction detailed above.
    Let $H$ denote a fixed graph with the properties from Lemma \ref{lem:existenceOfGridGraphs} for values $(d-1), n$.
     Now $G_{d,n}$ consists of $n+1$ disjoint copies of $H$, denoted $H_1, \dots, H_{n+1}$. We also add connections between the $H_i$ such that for each $i \in [n+1]$, every vertex in $H_i$ has precisely one edge leading into a different copy $H_j$. 
     Formally, name the vertices in the two parts of the bipartition of $H_i$ as follows: 
     $
        V(H_i) = \{v_{ij} \mid j \in [n+1] \setminus \{i\}\} \uplus \{w_{ij} \mid j \in [n+1] \setminus \{i\}\}.
     $
     Then let
     $
     E(G_{d,n}) \coloneqq \bigcup_{i \in [n+1]} E(H_i) \cup \{ v_{ij}w_{ji} \mid i,j \in [n+1], i \neq j \}.
     $
     Note that $G_{d,n}$ is again bipartite with the $v$-vertices and the $w$-vertices forming the two parts of the bipartition. Moreover, $G_{d,n}$ is $d$-regular because $H$ is $(d-1)$-regular.
     For values $n \in \bbN$ for which Lemma \ref{lem:existenceOfGridGraphs} does not guarantee the existence of a graph, $G_{d,n}$ is simply undefined.

    For the proofs of Theorems \ref{thm:mainAritySeparation} and \ref{thm:hierarchyForFixedL}, we use the following classes of structures, which will be shown to be inseparable in the respective logics of interest. For Theorem \ref{thm:mainAritySeparation}, for each $\ell \geq 3$, the two classes are $K_\ell \coloneqq \{ \StructA(G_{\ell,n}) \mid n \in \bbN  \}$ and $\tilde{\Kk}_\ell \coloneqq \{ \tilde{\StructA}(G_{\ell,n}) \mid n \in \bbN  \}$. For Theorem \ref{thm:hierarchyForFixedL}, for $r \geq \ell \geq 3$, the classes are $K^*_{r,\ell} \coloneqq \{ \StructA^*(G_{2r+1,n}) \mid n \in \bbN  \}$ and $\tilde{\Kk}^*_{r,\ell} \coloneqq  \{ \tilde{\StructA}^*(G_{2r+1,n}) \mid n \in \bbN  \}$.

    \subsection{(In-)distinguishability of the instances}
    We first show that there is a sentence in $\Ll(\Qq_{\ell+1}^{\NCSP^{\ell}})$ that separates $\Kk_\ell$ from $\tilde{\Kk}_\ell$.
    \begin{lemma}
    \label{lem:satisfiabilityOfCSP}
    For each $n \in \bbN$, there exists a homomorphism $h \colon \StructA(G_{\ell,n}) \to \StructB_\ell$. 
    \end{lemma}
    \begin{proof}
    Let $G_{\ell,n} = (V,E)$.
    A homomorphism $h \colon \StructA(G_{\ell,n}) \to \StructB_\ell$ can be defined by letting $h(e,a) = a$ for each $e \in E, a \in \bbZ_3$, and likewise, $h(v,a) = a$ for every $v \in V$. By definition of the relations $R_j$ in $\StructA(G_{\ell,n})$ and $\StructB_\ell$, this is a homomorphism. 
    \end{proof}

    \begin{lemma}
    \label{lem:nonSatisfiabilityOfCSP}
    For each $n \in \bbN$, there exists no homomorphism from $\tilde{\StructA}(G_{\ell,n})$ to $\StructB_\ell$. 
    \end{lemma} 
    \begin{proof}
    Again, let $G_{\ell,n} = (V,E)$.
    Consider the following system of linear equations $\Ee$ over $\bbZ_3$: For each $e \in E$, it has one variable $x_e$, and for each $v \in V \setminus \{\tilde{v}\}$, it has the equation $\sum_{e \in E(v)} x_e = 0$. The equation associated with $\tilde{v}$ reads $\sum_{e \in E(\tilde{v})} x_e = 1$. 
    Now suppose for a contradiction that there exists a homomorphism $h \colon \tilde{\StructA}(G_{\ell,n}) \to \StructB_{\ell}$. Then $h$ defines a solution $\lambda$ for $\Ee$ via $\lambda(x_e) \coloneqq h(e,0)$. 

    But such a solution cannot exist: Let $S \uplus T = V(G)$ be the bipartition of $G_{\ell,n}$. We assume w.l.o.g.\ that $\tilde{v} \in S$. Then we have $\sum_{v\in S} \sum_{e \in E(v)}\lambda(x_e) = 1$ and $\sum_{v\in T} \sum_{e \in E(v)}\lambda(x_e) =0$. However, this is a contradiction since clearly $\sum_{v \in S} \sum_{e \in E(v)}\lambda(x_e)=\sum_{e\in E}\lambda(x_e)=\sum_{v\in T} \sum_{e \in E(v)} \lambda(x_e)$.
   \end{proof} 

    
	\begin{corollary}
		\label{cor:CFIseparableWithCSPquantifier}
		There is a sentence $\psi \in \Ll(\Qq_{\ell+1}^{\NCSP^{\ell}})$ such that for each $n \in \bbN$, $\StructA(G_{\ell,n})$ and $\tilde{\StructA}(G_{\ell,n})$ are separated by $\psi$.
	\end{corollary}
	\begin{proof}
	By Lemma \ref{lem:modifiedTemplatePolyClosed}, $n^\ell$ is a polymorphism of $\StructB_\ell$.
    Thus, the CSP quantifier for the $(\ell+1)$-ary template $\StructB_\ell$ is in $\Qq_{\ell+1}^{\NCSP^{\ell}}$. 
    By Lemmas \ref{lem:satisfiabilityOfCSP} and \ref{lem:nonSatisfiabilityOfCSP}, this quantifier distinguishes $\StructA(G_{\ell,n})$ and $\tilde{\StructA}(G_{\ell,n})$, for every $n \in \bbN$.
	\end{proof}		
    The preceding assertions can be shown in exactly the same way for the classes $\Kk^*_{r,\ell}$ and $\tilde{\Kk}^*_{r,\ell}$, where the polymorphism-closedness of the template $\StructB^*_{r,\ell}$ is due to Lemma \ref{lem:modifiedTemplate2PolyClosed}. 

    To finish the proof of the separation $\Ll(\Qq_{\ell+1}^{\NCSP^{\ell}})\nleq \Ll(\Qq_{\ell-1})$ (and thus of Theorem \ref{thm:mainAritySeparation}), we show that no sentence in $\Ll(\Qq_{\ell-1})$ can distinguish between the classes $\Kk_\ell$ and $\tilde{\Kk}_\ell$. 
    To achieve this, we show that for every $k \geq \ell-1$,
    there exists $n \in \bbN$ such that Duplicator has a winning strategy in the game $\BP^k_{\ell-1}(\StructA(G_{\ell,n}),\tilde{\StructA}(G_{\ell,n}))$.
    Then it follows with Corollary \ref{cor:DuplicatorStrategyImpliesInseparatbilityOfClasses} that $\Kk_\ell$ and $\tilde{\Kk}_\ell$ cannot be separated by a sentence of $\Ll(\Qq_{\ell-1})$.

    The indistinguishability of $\Kk^*_{r,\ell}$ and $\tilde{\Kk}^*_{r,\ell}$ in $\Ll(\Qq_r)$ is shown with a similar argument. 
    So from now on, let $k \geq \ell-1$ be a fixed number of pebbles.
    We choose $n \in \bbN$ so that $G_{\ell,n}$ is defined, and $n$ is sufficiently larger than $k$. In the following, write $G \coloneqq G_{\ell,n}$ and let $V$ and $E$ be the vertex and edge set of $G$. Recall that $G$ is the disjoint union of identical graphs $H_1, \dots H_{n+1}$ with certain connections between them.
    \begin{lemma}
    \label{lem:duplicatorWins}
    Duplicator has a wining strategy in the game $\BP^k_{\ell-1}(\StructA(G),\tilde{\StructA}(G))$.
    \end{lemma}
The lemma is proved by showing that Duplicator can maintain an invariant which prevents Spoiler from winning.
Let $(\alpha,\beta)$ be a position in the game $\BP^k_{\ell-1}(\StructA(G),\tilde{\StructA}(G))$ with $\dom{\alpha}=\dom{\beta} \subseteq \{x_{1},\ldots,x_{k}\}$. For the sake of  simplicity, we denote the tuples of \emph{pebbled elements} in the image of $\alpha$ and $\beta$ by $\bar{\alpha}$ and $\bar{\beta}$, respectively. Furthermore, we call $(\bar{\alpha}, \bar{\beta})$ the corresponding \emph{pebble position} of the game.

A vertex $u\in V(G)$ is \emph{safe} in a pebble position $(\bar{\alpha}, \bar{\beta})$ if $u=v_{ij}$ or $u=w_{ij}$ and 
$A_v$ is pebble-free (i.e., no component of $\bar{\alpha}$ or $\bar{\beta}$ is in $A_v$) for every $v\in V(H_i)\cup V(H_j)$.

Let $u \in V$. We call a bijection $f \colon A_G \to A_G$ \emph{good bar $u$} if for every $v \in V \setminus \{u\}$, $\restrict{f}{A_v} \colon A_v \to A_v$ is a partial isomorphism $\StructA(G)\to\tilde{\StructA}(G)$ (thus, $f$ is an isomorphism between $\StructA(G)$ and $\tilde{\StructA}(G)$ except at $A_u$).

To prove Lemma \ref{lem:duplicatorWins}, we show that Duplicator can maintain the following \textbf{invariant} for the pebble positions $(\bar{\alpha}, \bar{\beta})$ reached by the moves of the players:
There exists a bijection $f \colon A_G \to A_G$ with $f(\bar{\alpha}) = \bar{\beta}$ that is good bar $u$, for some $u \in V$ which is \emph{safe}.

\begin{restatable}{lemma}{invariantCanBeMaintained}
\label{lem:invariantCanBeMaintained}
Let $(\bar{\alpha}, \bar{\beta})$ be the current pebble position in $\BP^k_{\ell-1}(\StructA(G),\tilde{\StructA}(G))$.
Assume the invariant holds in the beginning of the round. Then Duplicator has a strategy to ensure the invariant also in the beginning of the following round.
\end{restatable}
A full proof of the lemma can be found in Appendix \ref{sec:appendixArityHierarchySection}. The idea is as follows.
Duplicator begins the round by playing the bijection $f \colon \StructA(G) \to \tilde{\StructA}(G)$ sending $\bar{\alpha}$ to $\bar{\beta}$ which exists because the invariant holds at the beginning of the round. Let $u$ be the vertex such that $f$ is good bar $u$. Now Spoiler modifies the position of up to $\ell-1$ pebbles in both structures, leading to a new pebble position $(\bar{\alpha}', \bar{\beta}')$. 
Duplicator has to find a new bijection $f'$ that satisfies the invariant again. This can be done using a pebble-free path $P$ from $u$ to another vertex $u'$ which is safe with respect to $(\bar{\alpha}', \bar{\beta}')$: Indeed, a standard property of CFI-like constructions is that inconsistencies in bijections can always be ``shifted'' along a path in the base graph. 
More precisely, for every path $P$, and each constant $c \in \bbZ_3$, there exists a bijection $g_{P,c} \colon \StructA(G) \to \tilde{\StructA}(G)$ that induces a local isomorphism everywhere except at the vertex gadgets belonging to the two endpoints $v_1, v_2$ of $P$: For each $j \in \{1,2\}$, $g_{P,c}$ induces an isomorphism from $\StructA(v_j,s_j)$ to $\StructA(v_j,s_j+c)$, where $s_j \in \bbZ_3$ is the charge associated with $v_j$ in $\StructA(G)$. 
In other words, for $P$ a pebble-free path from $u$ to some safe $u'$, and for the appropriate choice of $c \in \bbZ_3$, $f' \coloneqq g_{P,c} \circ f$ will be good bar $u'$.
It remains to argue why such a pebble-free path to a safe vertex always exists. 
This can be done exactly like in \cite{Hella96}: Because the number of copies of $H_i$ in $G$ is sufficiently larger than $k$, there always exists some pebble-free copy $H_i$ and safe $u' \in V(H_i)$. The pebble-avoiding path $P$ from $u$ to $u'$ exists because each copy $H_i$ is sufficiently connected. 
The only difference to \cite{Hella96} is that now, if a pebble has been placed on $C_u = \{u\} \times \bbZ_3$, one has to take extra care: If the path $P$ starts in $u$ via the edge $e_2$ or $e_3 \in \bar{e}(u)$, then there is a problem because $\bar{u}_i \neq 0$ for $i \in \{2,3\}$ (see Definition \ref{def:template1} for the vector $\bar{u}$). As a consequence of this, $g_{P,c}$ cannot be defined so that it fixes the pebbled elements in $C_u$. The workaround in this case is to use two paths $P_1, P_2$ from $u$ to the same safe vertex $u'$, one via $e_2$ and one via $e_3$, and to suitably distribute the value $c \in \bbZ_3$ by which $f$ has to be corrected at $u$ across the two paths.\\

The \textbf{proof of Theorem \ref{thm:mainAritySeparation}} is now straightforward: Corollary \ref{cor:CFIseparableWithCSPquantifier} shows that for every fixed $\ell \geq 3$, our constructed classes $\Kk_\ell$ and $\tilde{\Kk}_\ell$ are separated by a sentence of $\Ll(\Qq_{\ell+1}^{\NCSP^\ell})$. Lemma \ref{lem:duplicatorWins} together with Corollary \ref{cor:DuplicatorStrategyImpliesInseparatbilityOfClasses} shows that the classes are inseparable in $\Ll(\Qq_{\ell-1})$. Thus, $\Ll(\Qq_{\ell+1}^{\NCSP^\ell}) \not\leq \Ll(\Qq_{\ell-1})$.
The \textbf{proof of Theorem \ref{thm:hierarchyForFixedL}} is analogous. We show that for every pebble number $k \geq r$, there is an $n \in \bbN$ such that Duplicator has a winning strategy in $\BP^k_{r}(\StructA^*(G_{2r+1,n}),\tilde{\StructA}^*(G_{2r+1,n}))$.
The invariant that Duplicator maintains in this game is the same as above. 
From the existence of the winning strategy it follows that $\Kk^*_{r,\ell}$ and $\tilde{\Kk}^*_{r,\ell}$ cannot be separated by a sentence in $\Ll(\Qq_r)$.
Separability of the classes with a CSP quantifier is shown similarly as in Corollary \ref{cor:CFIseparableWithCSPquantifier}.

\section{Limitations of partial-Maltsev-closed quantifiers}\label{sec:Maltsev}

Write $\Qq_{r}^{\Mm}$ for the class of all $r$-ary quantifiers that are closed under the partial Maltsev operation.   Note that the Maltsev operation is a fixed operation of arity three, and thus the arity of the polymorphism is not a variable that plays a role in out study.  Instead, we establish a hierarchy based on the arity of the quantifiers and also construct an example that separates the logic with $r$-ary Maltsev-closed quantifiers from the logic with all $r$-ary generalized quantifiers, at least in the context of a fixed number of variables.  The interest in this lies in the novel nature of the construction and also the novelty in the use of a Spoiler-Duplicator game specific to Maltsev-closed quantifiers.  In the present section we give a brief summary of the results.  Full details of the construction and proofs can be found in Appendix~\ref{app:Maltsev}.

We obtain the arity hierarchy for Maltsev-closed quantifiers by noting that the quantifiers used to establish the arity hierarchy in~\cite{Hella96} are Maltsev-closed.  To be specific, define for each $r \geq 3$ the structure $\StructB_r = (\{0,1\}, R_0,R_1)$ where $R_i = \{(a_1,\ldots,a_r) \mid \sum_{1\leq j \leq r} a_j = i \pmod 2\}$.  Then, it follows from the proof of \cite[Theorem 8.6]{Hella96} that $\CSP{\StructB_{r+1}}$ is not definable in $\Ll(\Qq_r^\Mm)$.  Since $\StructB_r$ admits a Maltsev polymorphism, the following is immediate.
\begin{restatable}{theorem}{maltsevHierarchy}
		\label{thm:maltsevHierarchy}
		For every $r \geq 3$,
		$
		\Ll(\Qq_r^\Mm) \lneq \Ll(\Qq^\Mm_{r+1}).
		$
	\end{restatable}	

Finally, we show that if we limit our logic to exactly $k$ variables, then there is a quantifier of arity $k$ that cannot be defined by using Maltsev-closed quantifiers of arity $k$, giving us the following theorem.
\begin{restatable}{theorem}{MaltsevSeparation}
		\label{thm:MaltsevSeparation}
		For every $k \geq 3$,
		$
		\Ll^k(\Qq_{k}^\Mm) \lneq \Ll^k(\Qq_{k}).
		$ 
\end{restatable}

The proof requires a novel construction and substantial technical work to establish. The details are given in  Appendix~\ref{app:Maltsev}. It leaves open the interesting question of whether we can show that this quantifier is not definable in $\Ll(\Qq_{k}^\Mm)$, without the limitation on the number of variables. The apparent difficulty in obtaining such inexpressibility results for Maltsev-closed quantifiers lies in the fact that the typical hard examples, which are based on the CFI-construction, cannot be used: Their templates admit a Maltsev polymorphism, and so they are distinguishable by a quantifier in $\Qq^\Mm$.
      	
\section{Conclusion}
Our main results show that the quantifiers closed under $\ell$-ary partial near-unanimity polymorphisms introduced at CSL'24 \cite{dawarHella2024} indeed give rise to new and interesting logics whose expressive power, in particular if $r=\ell+1$, is strictly different from any of the logics $\Ll(\Qq_r)$ that were studied by Hella in 1996 \cite{Hella96}. 

This establishes the logics $\Ll(\Qq_r^{\Nn^\ell})$ as formalisms whose finite-model-theoretic properties deserve further study. One question that comes to mind is what can be said about the equivalence relation induced by $\Ll(\Qq_r^{\Nn^\ell})$ on finite structures. Can it be characterized as a \emph{homomorphism indistinguishability relation} or via a \emph{game comonad} \cite{abramsky2017pebbling}? For the logics $\Ll(\Qq_r)$, the answer is affirmative \cite{oConghaile}, but for example, the invertible-map equivalences (which also stem from certain generalized quantifiers strictly between $\Qq_r$ and $\Qq_{r+1}$), there provably is no such characterization \cite{lichter_et_alCSL2024}, so it would be interesting to clarify the situation for $\Ll(\Qq_r^{\Nn^\ell})$. It would also be desirable to tighten the gap between the arity reduction we have in Theorem~\ref{thm:mainArityComparison} and the 
quadratic arity increase in Theorem \ref{thm:hierarchyForFixedL}: It seems plausible that already a  smaller increase in arity should give strictly more expressive power.

The long-term goal of the study of polymorphism-closed quantifiers is, however, to establish inexpressibility results for quantifiers stronger than $\Qq_r^{\Nn^\ell}$, in particular the partial Maltsev-closed ones. Theorem \ref{thm:MaltsevSeparation} is a first (and highly non-trivial) step in that direction but ultimately, we would like to exhibit an $r$-ary class of structures that is undefinable in the full logic $\Ll(\Qq_{r}^\Mm)$, without a restriction on the number of variables. 

An even harder open problem is to understand the expressive power of quantifiers closed under partial \emph{weak} near-unanimity polymorphism, for which we do not even know a model-comparison game yet. From a CSP perspective, weak near-unanimity is the most interesting condition as it characterizes precisely the finite-domain CSPs that are solvable in polynomial time (assuming P$\neq$NP).

\bibliographystyle{plainurl}	
\bibliography{references.bib}

	\newpage
	\appendix

    \section{Details omitted from Section \ref{sec:collapses}}
    \label{sec:appendixCollapseSection}
    \hypergraphColouring*
	\begin{proof}
          Given $\StructA$ in the vocabulary $\{R\}$ with $\ar(R)=r$, let $\Gg(\StructA)$ be the graph (possibly with self-loops) on vertex set $A$ with edges $\{(a_i,a_j)\mid (a_1,\ldots,a_r)\in R^\StructA, i\not=j\}$. Clearly $\Gg(\StructA)$ is $\FO$-definable in $\StructA$. Now we claim that $\StructA \to H^r_m$ if, and only if, $\Gg(\StructA)\to K_m$. 
		
		If $h: \StructA \to H^r_m$ is a homomorphism, then $h$ is also a homomorphism from $\Gg(\StructA)$ to $K_m$ because if $(a,b)$ is an edge in $\Gg(\StructA)$, then there is a tuple in $R^{\StructA}$ in which $a$ and $b$ appear together. Thus, $h$ maps $a$ and $b$ to different vertices in $[m]$, as desired. 
		
		Conversely, let $h\colon \Gg(\StructA) \to K_m$ be a homomorphism, and let $\bar{a} \in R^{\StructA}$.  Then the components of $\bar{a}$ are all distinct because $K_m$ has no self-loops, so otherwise, this homomorphism $h$ would not exist. In $\Gg(\StructA)$, the components of $\bar{a}$ form a clique, and so $h$ maps them all to distinct vertices in $[m]$, and hence $h(\bar{a})\in R^r_m$.
	\end{proof}

    \arityReductionNonCSP*
    \begin{proof}
          We start by defining the reduction $f$. For any $\tau$-structure $\StructA$,  note that $\StructA\le n^\ell(\StructA)$.  Thus, iterating the operation $n^{\ell}$ results, in a finite number of steps in a structure $\StructA^+$  with $\StructA\le\StructA^+$ and $n^{\ell}(\StructA^+) = \StructA^+$, i.e.\ $n^\ell$ is a partial polymorphism of $\StructA^+$.  It is easily seen that this process can be defined as an inflationary fixed-point and thus $\StructA^+$ is $\Ll$-definable in $\StructA$. 
	
	Let $(i_1,j_1), (i_2,j_2),..., (i_{\ell},j_\ell)$ be as in the proof of Corollary~\ref{cor:arityReductionQuantifiers}. We define now $f(\StructA)\coloneqq (\StructA^+)^*$, where the operator $\StructC\mapsto\StructC^*$ is as defined before Lemma~\ref{lem:CSPreductionByOne}. Note that $f$ is $\Ll$-definable as the composition of an $\Ll$-definable and $\FO$-definable reduction. 
	Furthermore, we define $\Kk'$ simply as $f(\Kk)$. Thus, the vocabulary $\sigma$ of $\Kk'$ is $\{R_n\mid R\in\tau,n\in [\ell]\}$. As in Corollary~\ref{cor:arityReductionQuantifiers}, the arity of all relations in $\sigma$ is at most $\lceil \frac{\ell-1}{\ell} \cdot r \rceil$.
	
		We still need to show that $\StructA \in \Kk$ if, and only if, $f(\StructA) \in \Kk'$. The only direction that requires work is to prove that $f(\StructA) \in \Kk'$ implies $\StructA \in \Kk$. Assume for this that	$f(\StructA) \in \Kk'$. Then there exists $\StructB \in \Kk$ such that $(\StructB^+)^*=(\StructA^+)^*$.  It follows now from Lemma \ref{lem:arityReductionTrick} that $\StructB^+=\StructA^+$: for any $R\in\tau$ and $\bar{a}\in A^r=B^r$ we have the following sequence of equivalences:
\begin{eqnarray*}
	\bar{a}\in R^{\StructA^+}&\iff& 
	p_{\setminus [i_n,j_n]}(\bar{a})\in R^{\StructA^+}_n \text{ for all } n\in [\ell] \text{ such that }i_n\le\ar(R)\\ &\iff& 			p_{\setminus [i_n,j_n]}(\bar{a})\in R^{\StructB^+}_n \text{ for all }n\in [\ell] \text{ such that }i_n\le\ar(R) \\ &\iff& 			\bar{a}\in R^{\StructB^+}.
\end{eqnarray*}
Since $\StructB\in\Kk$ and $\Kk$ is $n^\ell$-closed, we have $\StructB^+\in\Kk$. But then also $\StructA\in\Kk$ because $\StructA\le\StructA^+=\StructB^+$ and $\Kk$ is downwards monotone.
	\end{proof}

	 \section{Details omitted from Section \ref{sec:ArityHierarchy} } 
     \label{sec:appendixArityHierarchySection}

    We begin with two auxiliary lemmas that do not appear in the main part of the paper but will be needed in some of the omitted proofs.
	\begin{lemma}
		\label{lem:atLeastThreeDistinguishedColumns}
		Let $\ell \geq 3$ and $r \geq \ell$. Let $A \in \bbZ_3^{\ell \times r}$ be a matrix that has the $\ell$-ary near-unanimity property. If $\widehat{n^\ell}(A)$ is not a row of $A$, then there are at least $\ell$ 
		columns in $A$ that are non-constant and have their minority entry in $\ell$ distinct rows.
	\end{lemma}
	\begin{proof}
		Suppose for a contradiction that there are at most $\ell-1$ rows of $A$ which contain the minority values of columns. 
		Then there is at least one row $A_{i-}$ which contains in each column the respective majority value of that column. So $\widehat{n^\ell}(A) = A_{i-}$.
	\end{proof}

	\begin{restatable}{lemma}{separatingVectorPair}
		\label{lem:separatingVectorPair}
			Let $\ell \geq 3$ and $r \geq \ell$. Let $A \in \bbZ_3^{\ell \times r}$ be a matrix with the $\ell$-ary near-unanimity property.
			Let $A_{-j_1}, A_{-j_2}$ be two non-constant columns of $A$ that have their respective minority entry in distinct rows.
			Let $\bar{u}_{11} \in \bbZ_3^{r}$ be the vector that has a $1$-entry at index $j_1$ and $j_2$, and is $0$ everywhere else. Let $\bar{u}_{12} \in \bbZ_3^{r}$ be the vector that has a $1$-entry at index $j_1$, a $2$-entry at index $j_2$, and is $0$ everywhere else. 
			Then the vector $A \cdot \bar{u}_{11}$ or the vector $A \cdot \bar{u}_{12}$ contains three distinct entries. 
	\end{restatable}
	\begin{proof}
		Let $a_{j_1}, a_{j_2} \in \bbZ_3$ denote the majority values in the respective columns $j_1, j_2$, and let $b_{j_1}, b_{j_2}$ be the respective uniquely different values in these columns. 
		By the assumption of the lemma, $b_{j_1}, b_{j_2}$ are in distinct rows, say $i_1, i_2 \in [\ell]$. 
		We have $a_{j_1} \neq b_{j_1}$ and $a_{j_2} \neq b_{j_2}$. There are two cases. Suppose first that $b_{j_1} = b_{j_2}$. 
		Let $v = a_{j_1}+a_{j_2}, w = b_{j_1} + a_{j_2}, x = a_{j_1}+b_{j_2}$. Note that these three entries appear in $A \cdot \bar{u}_{11}$ in rows $i_1, i_2$, and any other row. Now we have $v \neq w$ (since $a_{j_1} \neq b_{j_1}$) and $v \neq x$ (since $a_{j_2} \neq b_{j_2}$). If $w \neq x$, then we are done. If $w = x$, then consider $v' = a_{j_1}+2a_{j_2}, w' = b_{j_1} + 2a_{j_2}, x' = a_{j_1}+2b_{j_2}$. These values appear in $A \cdot \bar{u}_{12}$. As before, $v' \neq w'$ and $v' \neq x'$ is easy to see. 
		Since $w=x$, it follows that $w' - x' = (w-x)+a_{j_2}-b_{j_2} = a_{j_2}-b_{j_2} \neq 0$. Thus, $w' \neq x'$, and so, $A \cdot \bar{u}_{12}$ contains three distinct entries.
	\end{proof}

    Now we are ready to present the proofs of the lemmas in Section \ref{sec:ArityHierarchy}. 	
    \separatingVector*
 	\begin{proof}
        By Lemma \ref{lem:atLeastThreeDistinguishedColumns}, all the $\ell$ columns of $A$ are non-constant and have their minority entry in distinct rows. Without loss of generality we can assume that the respective minority entries are on the diagonal (otherwise, reorder the rows).
		Let $A' \in \bbZ_3^{3 \times 3}$ be the $3 \times 3$-submatrix of $A$ in the top left corner. It is easy to see that $A'$ satisfies conditions 1) and 3) from Lemma \ref{lem:baseCaseSeparatingVector}. It also satisfies condition 2) from that lemma because the row sums of $A$ are all equal by assumption 2), and in $A$, the first three rows are identical in all columns except the first three.
		Therefore, the first three entries of the vector $A \cdot (0,1,2,0,0,...,0)^T$ are all distinct by Lemma \ref{lem:baseCaseSeparatingVector}. 
	\end{proof}

\conditionForGettingARowBack*
\begin{proof}
    Since $A$ has the $\ell$-near-unanimity property by assumption, $\widehat{n^\ell}(A)$ is defined. If $A$ has at least two rows that are identical, then $\widehat{n^\ell}(A)$ must be that row because for each $j \in [\ell]$, the majority value in the $j$-th column of $A$ is equal to the $j$-th entry of the row that appears twice.  
 
	Conversely, assume that the polymorphism returns a row $\bar{a}$ of $A$. Then every entry of $\bar{a}$ is the majority one in its column. Assume for a contradiction that row $\bar{a}$ occurs only once in $A$. Then for every other row $\bar{a}'$ of $A$, there exists a column $j \in [r]$ such that $\bar{a}_j \neq \bar{a}'_j$. 
	But this cannot be the same $j \in [\ell]$ for more than one row $\bar{a}'$ because then, the column $j$ would violate the $\ell$-near-unanimity property. 
	Therefore, we can define a submatrix $A' \in \bbZ_3^{(\ell-1) \times (\ell-1)}$ which arises from $A$ by removing $\bar{a}$ and a column $j$ so that up to reordering of rows, in $A'$, each column contains precisely one value that is different from the others, and this is on the diagonal. 
	This also means that each row of $A'$ is identical to $\bar{a}$ except at precisely one position. Hence, the sum over any row of $A'$ differs from $\sum \bar{a}$. In the removed column $j$, all entries except possibly one are equal to $\bar{a}_j$. So in each row of $A$ except potentially one, the row sum is not equal to $\sum \bar{a}$ because the rows differ at exactly one position. Hence the second assumption is violated, and this case cannot occur. Thus, there is at least one other row of $A$ that is equal to $\bar{a}$. 
\end{proof}

    \separatingVectorPairs*
    \begin{proof}
		Split $[r]$ into sets $K_1 \uplus\cdots \uplus K_{\ell-1} = [r]$ such that $|K_1|=\cdots=|K_p|=q+1$ and $|K_{p+1}|=\cdots=|K_{\ell-1}|=q$. 
		Define $\Uu$ as follows. For every $2$-element set of columns $\{j_1,j_2\}$ such that $\{j_1,j_2\} \subseteq K_i$ for some $i\in [\ell-1]$, include in $\Uu$ the respective vectors $\bar{u}_{11}, \bar{u}_{12}$ given by Lemma~\ref{lem:separatingVectorPair} for this choice of $j_1, j_2$.
		Note that $|\Uu| = 2 \bigl(p \binom{q+1}{2}+(\ell-1-p)\binom{q}{2}\bigr) = 2 \bigl( p\frac{(q+1)q}{2}+(\ell-1-p)\frac{q(q-1)}{2}\bigr)=(\ell-1)q(q-1)+2pq$.
		
		Assume first that $\widehat{n^\ell}(A)$ is not a row of $A$.
        By Lemma~\ref{lem:atLeastThreeDistinguishedColumns}, $A$ has $\ell$ columns that have a distinguished entry which is in a different row, respectively. So there is an $i \in [\ell-1]$ such that $K_i$ contains at least two of these columns, call them $j_1, j_2$. So by definition of $\Uu$, the vectors $\bar{u}_{11}, \bar{u}_{12}$ for $j_1, j_2$ are in $\Uu$, and by Lemma~\ref{lem:separatingVectorPair}, one of them satisfies that its product with $A$ contains three distinct entries. 
        
        Now assume that $\widehat{n^\ell}(A) = \bar{a}$ is a row of $A$. Let $\bar{u} \in \Uu$. Suppose  $\widehat{n^\ell}(A \cdot \bar{u})$ is defined, i.e.\ the vector $A \cdot \bar{u}$ is constant or has at most one entry that is different from the others. We want to show that the value $\bar{a} \cdot \bar{u}$ appears at least twice in the vector $A \cdot \bar{u}$ -- then it follows that $\widehat{n^\ell}(A \cdot \bar{u}) = \bar{a} \cdot \bar{u}$. Let $j_1, j_2$ be the two columns of $A$ where $\bar{u}$ is non-zero. If $\ell \geq 4$, then there are at least two distinct rows $i_1, i_2$ of $A$ such that $A_{i_1j_1} = A_{i_2j_1} = \bar{a}_{j_1}$, and $A_{i_1j_2} = A_{i_2j_2} = \bar{a}_{j_2}$ (because $\bar{a}_{j_1}$ and $\bar{a}_{j_2}$ are the respective majority entries in their columns, and all but at most one entry in the column must have this value). This means that $\bar{a} \cdot \bar{u}$ appears in row $i_1$ and in row $i_2$ of the vector $A \cdot \bar{u}$, which is at least twice, as desired. In case that $\ell = 3$, it can happen that there are no two distinct rows $i_1, i_2$ such that $A_{i_1j_1} = A_{i_2j_1} = \bar{a}_{j_1}$, and $A_{i_1j_2} = A_{i_2j_2} = \bar{a}_{j_2}$. But then w.l.o.g.\ the situation is this: $A_{1j_1} \neq \bar{a}_{j_1}$ is the minority value in column $j_1$, $A_{2j_2} \neq \bar{a}_{j_2}$ is the minority value in column $j_2$, and $A_{3j_1} = \bar{a}_{j_1}$, $A_{3j_2} = \bar{a}_{j_2}$ are the respective majority entries in the columns. Then by Lemma \ref{lem:separatingVectorPair}, one of the vectors $A \cdot \bar{u}_{11}$ and $A \cdot \bar{u}_{12}$ for the columns $j_1, j_2$ has three distinct entries. Since these $\bar{u}_{11}, \bar{u}_{12}$ are in $\Uu$, there exists a $\bar{u} \in \Uu$ such that $A \cdot \bar{u}$ has three distinct entries, and this is the other way in which the second statement of the lemma can be satisfied.
	\end{proof}

\existenceOfGridGraphs*
\begin{proof}
    In case that $d$ is even, a $(d/2)$-dimensional toroidal grid of suitable size satisfies these properties. If $d$ is odd, then we can take two copies of a $(d-1)/2$-dimensional toroidal grid and join every vertex with its copy to achieve degree $d$ at every vertex. 
\end{proof}

\paragraph*{Duplicator winning strategy}
We present the details that are needed to formally prove Lemma \ref{lem:duplicatorWins}. We first of all consider properties of bijections on the gadgets $\StructA(v,s)$. Fix a graph $G=(V,E)$ as the base graph.
Let $v\in V$ with $\bar{e}(v)=(e_1,\ldots,e_{\ell})$. We say that a bijection $f\colon A_v\to A_v$ is \emph{edge-preserving} if $\{f(e_i,0),f(e_i,1),f(e_i,2)\}=\{(e_i,0),(e_i,1),(e_i,2)\}$ for all $i\in [\ell]$, and $\{f(v,0),f(v,1),f(v,2)\}$\\ $=\{(v,0),(v,1),(v,2)\}$. Furthermore, an edge-preserving $f$ is \emph{cyclic} if for each $i\in [\ell]$, there is $c_i \in\bbZ_3$ such that for all $a\in\bbZ_3$,
	\begin{itemize}
	\item $f(e_i,a)=(e_i,a+c_i)$,
	\item $f(v,a)=(v,a+c_s+2\cdot c_t)$, where $\bar{u}_{s}=1$ and $\bar{u}_{t}=2$ are the nonzero components of $\bar{u}$. 
	\end{itemize}
	The \emph{shift value} of a cyclic bijection $f$ is $c(f)\coloneqq\sum_{i\in [\ell]}c_i$, and its \emph{shift on edge $e_i$} is $c_i$.
	We make the following observation:
	
	\begin{lemma}\label{lem:CompIsomAutom}
	Let $f\colon A_v\to A_v$ be a cyclic bijection. Then for all $s\in\bbZ_3$
	\begin{enumerate}[label=(\alph*)]
	\item $f$ is an automorphism of $\StructA(v,s)$ if and only if $c(f)=0$. 
	\item $f$ is an isomorphism $\StructA(v,s)\to\StructA(v,s+1)$ if and only if $c(f)=1$. 
	\item $f$ is an isomorphism $\StructA(v,s+1)\to\StructA(v,s)$ if and only if $c(f)=2$. 	
    \end{enumerate}
    \end{lemma}

In order to show that Duplicator can always keep the mistake in the bijection away from the currently pebbled elements, we need to prove a typical property of CFI-like constructions:
Given two vertices $v,w \in V$, and a path $P$ between them in $G$, we can always find a bijection between $\StructA(G)$ and $\tilde{\StructA}(G)$ that defines an isomorphism from $\StructA(v,s)$ to $\StructA(v,s+c)$, for any $c \in \bbZ_3$, and an isomorphism from $\StructA(w,t)$ to $\StructA(w,t+c)$ (or, actually, to $\StructA(w,t-c)$, if $P$ has even length), and an automorphism of every other vertex gadget $\StructA(v',s')$:
    \begin{lemma}
    \label{lem:pathIso}
    Let $G = (V,E)$ be a connected regular graph, and let $v,w \in V$ be distinct vertices, and $P$ a path from $v$ to $w$ in $G$.
    Let $c \in \bbZ_3$.
    There exists a bijection $g_{P,c} \colon A_G \to A_G$ such that:
    \begin{enumerate}
    \item $g_{P,c}$ is the identity everywhere except on $P$, that is, for each $v' \in V$ not on the path $P$, $g_{P,c}(v',a) = (v'a)$ and $g_{P,c}(e,a) = (e,a)$ for each $(v',a)$ and $(e,a)$ in $A_{v'}$.
    \item Let $s \in \bbZ_3$ be the value such that $\StructA(v,s)$ is the vertex gadget associated with $v$ in $\StructA(G)$. Then $\restrict{g_{P,c}}{A_{v}}$ is an isomorphism from $\StructA(v,s)$ to $\StructA(v,s+c)$.
    \item Let $t \in \bbZ_3$ be the value such that $\StructA(w,t)$ is the vertex gadget associated with $w$ in $\tilde{\StructA}(G)$. Then $\restrict{g_{P,c}}{A_{w}}$ is an isomorphism from $\StructA(w,t)$ to $\StructA(w,t+c)$ if the number of edges in $P$ is odd; else, it is an isomorphism from $\StructA(w,t)$ to $\StructA(w,t-c)$.
    \item For every interior vertex $v'$ on the path $P$, $\restrict{g_{P,c}}{A_{v'}}$ is an automorphism of $\StructA(v',x)$, where $x$ is the value associated with $\StructA(v',x)$ in $\StructA(G)$.
    \end{enumerate}
    \end{lemma}
    \begin{proof}
    We define $g_{P,c}$ individually on the $A_{v'}$, for each $v' \in V$. For each $v'$ that is not on the path $P$, we define $\restrict{g_{P,c}}{A_{v'}} \colon A_{v'} \to A_{v'}$ as the identity map. Thus, it satisfies the first claimed property.

    Now consider the first endpoint of $P$, which is $v$. 
    Let $e \in E(v)$ be the first edge in $P$. We define $\restrict{g_{P,c}}{A_{v}} \colon A_{v} \to A_{v}$ as the cyclic bijection with shift $0$ on each edge except $e$, where it has shift $c$. By Lemma \ref{lem:CompIsomAutom}, this satisfies the second claimed property.

    Let $v'$ be the other endpoint of $e$. If $v' \neq w$, then 
    we define $\restrict{g_{P,c}}{A_{v'}} \colon A_{v'} \to A_{v'}$ as a cyclic bijection with shift value $0$. Let $e' \in E(v')$ be the second edge on $P$. 
    To ensure that $\restrict{g_{P,c}}{A_{v'}}$ has shift value $0$, we define it such that its shift on $e'$ is $-c$, and it is the identity on all edges in $E(v') \setminus \{e,e'\}$.

    In this way, we continue the definition of $g_{P,c}$ along the path. It is easy to see that by Lemma \ref{lem:CompIsomAutom}, all claimed properties hold.
    \end{proof}

With this, we can prove the key lemma from which it follows immediately that Duplicator has a winning strategy in the game $\BP^k_{\ell-1}(\StructA(G),\tilde{\StructA}(G))$, where $G \coloneqq G_{\ell,n}$ is the graph that we constructed, consisting of identical copies of graphs $H_1, \dots, H_{n+1}$. For the definition of the invariant that Duplicator maintains, we refer back to the text before Lemma \ref{lem:invariantCanBeMaintained} in the main part of the paper.

\invariantCanBeMaintained*
\begin{proof}
Duplicator begins the round by playing the bijection $f \colon \StructA(G) \to \tilde{\StructA}(G)$ sending $\bar{\alpha}$ to $\bar{\beta}$, which exists because the invariant holds at the beginning of the round. Let $u$ be the vertex such that $f$ is good bar $u$. Now Spoiler modifies the position of up to $\ell-1$ pebbles in both structures, leading to a new pebble position $(\bar{\alpha}', \bar{\beta}')$. 
Duplicator has to find a new bijection $f'$ that satisfies the invariant again. This can be done using a pebble-free path $P$ from $u$ to another vertex $u'$ which is safe with respect to $(\bar{\alpha}', \bar{\beta}')$. Then essentially $f'$ can be obtained from $f$ by composing it with a bijection as given by Lemma \ref{lem:pathIso} for $P$. The details are as follows:
We make a case distinction according to whether a pebble of $\bar{\alpha}'$ lies on $C_u = \{u\} \times \bbZ_3$ or not.\\
\\
\textit{Case 1:}\\
Suppose that $C_u$ is pebble-free. 
Then the argument proceeds exactly like in \cite{Hella96}: 
There are two subcases. The first is that in $\bar{\alpha}'$, there are still less than $\ell-1$ pebbles lying on vertices in 
\[
W \coloneqq \Big(\bigcup_{v \in V(H_i)} A_v\Big) \setminus (C_u \cup \{ (e^{ij},0), (e^{ij},1), (e^{ij},2) \}),  
\]
where $e^{ij} \in E$ denotes the edge connecting $u \in V(H_i)$ with its neighbour in $V(H_j)$.
Because $H_i$ is $(\ell-1)$-connected by definition, there exists a pebble-free path\footnote{a path $P$ is pebble-free if neither on $\{v\} \times \bbZ_3$, for a vertex $v$ in $P$, nor on $\{e\} \times \bbZ_3$, for an edge $e$ in $P$, there is a pebble} from $u$ to any other vertex of $H_i$. Moreover, because $n$ is big enough, there is some $i' \in [n+1]$ such that $H_{i'}$ is pebble-free and there exists a \emph{safe} vertex $u' \in H_{i'}$. Hence, we find a pebble-free path $P$ from $u$ to this $u' \in V(H_{i'})$. 

The second subcase is that all the $\ell-1$ pebbles that Spoiler has moved in this round are now on elements of $W$. But then, $H_j$ and the connecting edge $e^{ij}$ are still pebble-free. Thus, we can again find a pebble-free path $P$ to some safe $u' \in H_{i'}$, namely via $e^{ij}$ through $H_j$ to $H_{i'}$. 

In each of the two subcases, we have obtained our desired path $P$ from $u$ to a safe $u'$. Let $s, t \in \bbZ_3$ be the values of the gadget of $u$ in $\StructA(G)$ and $\tilde{\StructA}(G)$, that is: $\StructA(u,s)$ is the gadget associated with $u$ in $\StructA(G)$, and $\StructA(u,s)$ the gadget associated with $u$ in $\tilde{\StructA}(G)$.
Let $c \in \bbZ_3$ be such that $\restrict{f}{A_u}$ is an isomorphism from $\StructA(u,s+c)$ to $\StructA(u,t)$ (in particular, $c \neq 0$ because $f$ is not good at $u$).
Let $g_{P,c} \colon \StructA(G) \to \tilde{\StructA}(G)$ be a bijection as given by Lemma \ref{lem:pathIso}.
Then set $f' \coloneqq g_{P,c} \circ f$. This finishes \textit{Case 1}.\\
\\
\textit{Case 2:}\\
In this case, a pebble of $\bar{\alpha}'$ lies on $C_u$. 
We can find a safe vertex $u'$ and a path $P$ from $u$ to $u'$ with exactly the same reasoning as in Case 1. If such a path $P$ exists that begins with an edge $e_d \in \bar{e}(u)$ such that $\bar{u}_d = 0$, then we can proceed exactly as before, and define $f' \coloneqq g_{P,c} \circ f$. This is because $g_{P,c}$ is the identity on $C_u$ in this case (see again the definition of \emph{cyclic} bijections for how they act on $C_u$). 

It remains to deal with the case that for every potential path from $u$ to a safe vertex, $\bar{u}_d \neq 0$ for the edge $e_d \in \bar{e}(u)$ that the path begins with. Note that $\bar{u}_d \neq 0$ only for $d \in \{2,3\}$.
In this case we argue as follows. 
One pebble has been moved to $C_u$ in this round, and another pebble must have been moved to $\{e^{ij}\} \times \bbZ_3$ or to its other endpoint in $H_j$: By the choice of the order $\leq$ on $G$, the edge $e^{ij}$ is in a position $d \notin \{2,3\}$ in the tuple $\bar{e}(u)$, so $\bar{u}_d = 0$. So if $e^{ij}$ were pebble-free, then contrary to the assumption of this subcase, we can find a path to a safe vertex (with the argument as in the second subcase of Case 1) that ``avoids'' the non-zero entries of the vector $\bar{u}$. So, at least two pebbles must have been spent on $C_u$ and $e^{ij}$, which leaves at most $\ell-3$ pebbles that can lie on vertices in $W$ now. 
The graph $H$ is $(\ell-1)$-connected, so there must be two distinct edges $e,e' \in E(u) \setminus \{e^{ij}\}$ via which every other vertex in $H_i$ can be reached from $u$ with a pebble-free path. 
These two edges must be $e_2, e_3$ in the tuple $\bar{e}(u)$ because else, we would again have a path that avoids the non-zero entries of $\bar{u}$. 

So we can choose two (not necessarily disjoint) pebble-free paths $P_1, P_2$ from $u$ to a safe vertex $u'$ (using the reasoning from Case 1, subcase 1), one beginning with $e_2$ and one beginning with $e_3$.

Now as before, let $c \in \bbZ_3$ be the constant such that $\restrict{f}{A_u}$ is an isomorphism from $\StructA(u,s+c)$ to $\StructA(u,t)$. Let $b_1, b_2 \in \bbZ_3$ be such that $b_1 + b_2 = c$ and $b_1 + 2b_2 = 0$. It is easy to verify that such $b_1, b_2$ exist for every $c$. 
Let $g_{P_1,b_1}, g_{P_2,b_2}$ be the respective bijections obtained with Lemma \ref{lem:pathIso}.
Set $f' \coloneqq g_{P_1, b_1} \circ g_{P_2, b_2} \circ f$.
It remains to check that $g_{P_1, b_1} \circ g_{P_2, b_2}$ is the identity map on $C_u$, which is required because at least one vertex in $C_u$ is now pebbled. This follows from the fact that $b_1 + 2b_2 = 0$, as can be seen from the definition of cyclic bijections.
\end{proof}

	\section{Limitations of partial-Maltsev-closed quantifiers}\label{app:Maltsev}
	Write $\Qq_{r}^{\Mm}$ for the class of all $r$-ary quantifiers that are closed under the partial Maltsev operation.

    We now present and analyze a construction of structures that are distinguishable with $r$-ary generalized quantifiers but not with $r$-ary Maltsev-closed ones. To be precise, our lower bound holds for the $k$-variable fragment of the logic, for $k = r$. For a logic $\Ll$, let $\Ll^k$ denote its $k$-variable fragment.

    \MaltsevSeparation*

    The structures we use in the proof are indistinguishable in $\Ll^k(\Qq_{k}^\Mm)$ but can possibly be distinguished in $\Ll(\Qq_{k}^\Mm)$, i.e.\ with more variables available. 

    \paragraph*{Constructing the instances}
 
    Throughout this section, all arithmetic operations are in the ring $\bbZ_4$, unless stated otherwise.  The structures we construct are over a vocabulary with two relations $R_0$ and $R_1$ of arity $k$, and one binary relation $\prec$.
    
As before, we start with a $k$-regular connected base graph $G=(V,E)$.
Again, we first define vertex gadgets $\StructA^\Mm(v,s)$, for $v \in V$, $s \in \{0,1\}$  and then put them together to obtain the whole structure.

\begin{definition}[Vertex gadgets, Maltsev counterxample] 
    \label{def:ComponentsMaltsev}
    Fix $k \geq 3$ and let $G = (V,E)$ be a connected $k$-regular graph with an order $\leq^G$ on its vertices.
	Let $\bar{e}(v) = (e_1, \dots , e_{k})$ be the tuple of edges incident to $v$, in the order induced by $\leq^G$. Then for each $s\in\{0,1\}$, we let $\StructA^\Mm(v,s)\coloneqq(A^\Mm_v,R_0^{\StructA^\Mm(v,s)},R_1^{\StructA^\Mm(v,s)})$, where
	\begin{itemize}
	\item $A^\Mm_v \coloneqq E(v)\times\bbZ_4$,
	\item $R_0^{\StructA^\Mm(v,s)}\coloneqq \{((e_1,a_1),\ldots,(e_{k},a_{k})) \mid (\sum_{i\in[k]}a_ i) - 2s \in \{0,1\} \pmod 4\}$
        \item $R_1^{\StructA^\Mm(v,s)}\coloneqq \{((e_1,a_1),\ldots,(e_{k},a_{k})) \mid (\sum_{i\in[k]}a_ i) - 2s \in \{2,3\} \pmod 4 \}$
	\end{itemize}
\end{definition}

We combine the vertex gadgets in the same way as before to obtain our indistinguishable instances:
\begin{definition}\label{def:CFIstructMaltsev}
	Let $G = (V,E,\leq^G)$ be an $\ell$-regular connected graph with a linear order on its vertices. 
	\begin{enumerate}[label=(\alph*)]
	\item For each $U\subseteq V$, we define\\ 
$\StructA^\Mm(G,U)= (A^\Mm_G,R_0^{\StructA^\Mm(G,U)},R_1^{\StructA^\Mm(G,U)},R_2^{\StructA^\Mm(G,U)}, \prec^{\StructA^\Mm(G,U)})$, where
	\begin{itemize}
	\item $A_G^\Mm \coloneqq \bigcup_{v\in V}A^\Mm_v$,
	\item $R_j^{\StructA^\Mm(G,U)}\coloneqq \bigcup_{v\in V\setminus U}R_j^{\StructA^\Mm(v,0)}\cup \bigcup_{v\in U}R_j^{\StructA^\Mm(v,1)}$ for each $j\in\{0,1\}$.
    \item $\prec^{\StructA^\Mm(G,U)}$ is a preorder with $(e,a) \prec (e',a')$ if and only if $e \leq e'$ in the lexicographic order on $E$ induced by $\leq^G$. 
	\end{itemize}
	
	\item Furthermore, we define $\StructA^\Mm(G)\coloneqq\StructA^\Mm(G,\emptyset)$ and 
	$\tilde{\StructA}^\Mm(G)\coloneqq\StructA^\Mm(G,\{\tilde{v}\})$, where $\tilde{v}\in V$ is the smallest vertex of $G$ with respect to the order $\leq^G$.
    \end{enumerate}
\end{definition}

We now aim to show two things.  The first is that $\StructA^\Mm(G)$ and $\StructA^\Mm(G,\{\tilde{v}\})$ are not isomorphic.  Since the vocabulary has relations of arity at most $k$, this immediately implies that there is a $k$-ary quantifier $Q$ that distinguishes them.  The second is that $\StructA^\Mm(G)$ and $\StructA^\Mm(G,\{\tilde{v}\})$ cannot be distinguished in a $k$-pebble bijective game for Maltsev-closed quantifiers.  This is the game defined in~\cite[Definition~16]{dawarHella2024}.  We reproduce the definition for completeness.  The following defines a game $\Mm_k(\StructA,\StructB,\alpha,\beta)$ played on structures $\StructA$ and $\StructB$ with $X$ a set of $k$ variables and partial assignments $\alpha$ and $\beta$ of values in $\StructA$ and $\StructB$ respectively to the variables in $X$.  Below, for a relation $P$, $m(P)$ refers to the closure of $P$ under the partial Maltsev operation.

\begin{definition}\label{modified-game} 
The rules of the game $\Mm_k(\StructA,\StructB,\alpha,\beta)$ 
are the following: 
\begin{enumerate}
\item If $\alpha\mapsto\beta$ is not a partial isomorphism, then the game ends, and Spoiler wins.

\item If (1) does not hold, there are two types of moves that Spoiler can choose to play: 
\begin{itemize}

\item {\bf Left $\Mm$-quantifier move:} Spoiler starts by choosing $r\in [k]$ 
and an $r$-tuple 
$\vec y\in X^r$ of distinct variables. 
Duplicator responds with 
a bijection $f\colon B\to A$. Spoiler answers by choosing an $r$-tuple 
$\vec b\in B^r$. Duplicator answers by choosing $P = \{\vec{a}_1,\vec{a}_2,\vec{a}_3\} \subseteq A^r$
such that $f(\vec b)\in m(P)$. 
Spoiler completes the round by choosing $\vec a\in P$.
The players continue by playing $\Mm_k(\alpha',\beta')$, 
where $\alpha':=\alpha[\vec a/\vec y]$ and $\beta':=\beta[\vec b/\vec y]$.

\item {\bf Right $\Mm$-quantifier move:} Spoiler starts by choosing $r\in [k]$ 
and an $r$-tuple 
$\vec y\in X^r$ of distinct variables. 
Duplicator chooses next 
a bijection $f\colon A\to B$. Spoiler answers by choosing an $r$-tuple 
$\vec a\in A^r$. Duplicator answers by choosing $P = \{\vec{b}_1,\vec{b}_2,\vec{b}_3\}\subseteq B^r$
such that $f(\vec a)\in m(P)$. 
Spoiler completes the round by choosing $\vec b\in P$.
The players continue by playing $\Mm_k(\alpha',\beta')$,
where $\alpha':=\alpha[\vec a/\vec y]$ and $\beta':=\beta[\vec b/\vec y]$.
\end{itemize}

\item Duplicator wins the game if Spoiler does not win it in a finite number of rounds.
\end{enumerate}
\end{definition}

By ~\cite[Theorem~17]{dawarHella2024}, Duplicator has a winning strategy in $\Mm_k(\StructA,\StructB,\alpha,\beta)$ if, and only if, there is no formula of $\Ll^k(Q^{\Mm}_k)$ that distinguishes $(\StructA,\alpha)$ from $(\StructB,\beta)$.

But first, the Duplicator winning strategy that we are going to define in the game involves playing bijections that do not necessarily permute edge vertices $\{e\} \times \bbZ_4$ cylically, as before. Now, \emph{reflections} also play a role. We develop the necessary background.

\paragraph*{Properties of bijections over $\bbZ_4$}

\begin{lemma}
\label{lem:permutationsOnZ4}
Let $\pi \colon \bbZ_4 \to \bbZ_4$ be any permutation of $\bbZ_4$. Then one of the following two cases applies:
\begin{enumerate}
\item There exists $x \in \bbZ_4$ such that $\pi(x+1) = \pi(x)+1$.
\item There exists $a \in \bbZ_4$ such that for every $x \in \bbZ_4$, $\pi(x) = a-x$.
\end{enumerate}
\end{lemma}
\begin{proof}
We show that if $\pi$ does not satisfy (1), then it must satisfy (2). Let $a \coloneqq \pi(0)$. Because (1) does not apply, we have $\pi(1) \in \{a+2, a+3\}$. Suppose $\pi(1) =a+2$. Then $\pi(2) = a+1$ and $\pi(3)=a+3$. But then $\pi(0) = \pi(3+1) = (a+3)+1 = \pi(3)+1$. So condition (1) is satisfied for $x = 3$. Therefore, we must have $\pi(1) = a+3$. Then $\pi(2) \in \{a+1, a+2\}$, but we cannot have $\pi(2) = a+1$ because then, $\pi(3) = a+2$, so we would be in case (1) again, for $x = 2$. Therefore, $\pi(2) = a+2$ and $\pi(3)= a+1$. That means, $\pi$ is exactly as described in case (2).
\end{proof}

We call permutations of the second kind \emph{reflections}.
A permutation $\pi \colon \bbZ_4 \to \bbZ_4$ for which there is a constant $a \in \bbZ_4$ such that for each $x \in \bbZ_4,$ $\pi(x) = x+a$, is called a \emph{rotation} (or \emph{cyclic}). 
For an edge $e \in E$, a bijection $f \colon \{e\} \times \bbZ_4 \to \{e\} \times \bbZ_4$ is said to be a rotation or a reflection if its induced permutation on $\bbZ_4$ has this property.

In what follows, for $v \in V, j \in \{0,1\}$, we use $R_j(v)$ as shorthand to denote the relation $R_j^{\StructA^\Mm(v,0)}$.

\begin{lemma}
\label{lem:localIsomorphisms}
Let $v \in V(G)$ and let $f$ be a bijection defined from the set $\bigcup_{e \in E(v)} \{e\} \times \bbZ_4$ to itself, such that $f(\{e\} \times \bbZ_4) = \{e\} \times \bbZ_4$ for each $e \in E(v)$.
Suppose $f$ is either a rotation on every $e \in E(v)$ or a reflection on every $e \in E(v)$. For each $e \in E(v)$, let $c_e \in \bbZ_4$ denote the \emph{shift value} of $f$ on $e$, i.e\ such that $f(e,x) = (e,x+c_e)$ or $f(e,x) = (e,-x+c_e)$, respectively, for each $x \in \bbZ_4$. Let $c \coloneqq \sum_{e \in E(v)} c_e$.
\begin{enumerate}
\item If $f$ is a rotation on every $e \in E(v)$, then $f(R_j(v)) = R_j(v)$ for each $j \in \{0,1\}$ if and only if $c = 0$. Moreover, $f(R_j(v)) = R_{1-j}(v)$ for each $j \in \{0,1\}$ if and only if $c = 2$.  
\item If $f$ is a reflection on every $e \in E(v)$, then $f(R_j(v)) = R_j(v)$ for each $j \in \{0,1\}$ if and only if $c = 1$. Moreover, $f(R_j(v)) = R_{1-j}(v)$ for each $j \in \{0,1\}$ if and only if $c = 3$.  
\end{enumerate}
\end{lemma}
\begin{proof}
Suppose first that $f$ is a rotation on every $e \in E(v)$. 
If $c \in \{1,3\}$, then we see that $f$ neither preserves, nor exchanges $R_0(v)$ and $R_1(v)$: If, for example, $c = 1$, then for $(e_1,\dots,e_k) = E(v)$, $f$ maps the tuples $((e_1,a_1),\dots,(e_k,a_k))$ with $\sum_{i \in [k]} a_i = 0$ to tuples whose sum is $1$ (keeping them in $R_0(v)$), but it also maps the tuples whose sum is $1$ to tuples with sum $2$ (moving them from $R_0(v)$ to $R_1(v)$). The situation is similar if $c=3$.
Hence, if $f$ maps $R_j(v)$ to $R_j(v)$ or to $R_{1-j}(v)$, then we must have 
$c \in \{0,2\}$. Clearly, if $c = 0$, then $f$ preserves $R_j(v)$ for each $j \in \{0,1\}$. It is also not hard to see that adding $2$ to the sum of any tuple moves it from $R_0(v)$ to $R_1(v)$ or vice versa. This shows the first part of the lemma.
The second part is easily seen with a similar analysis.
\end{proof}

\begin{lemma}
\label{lem:notLocalIsomorphisms}
Let $v \in V(G)$ and let $f$ be as in Lemma \ref{lem:localIsomorphisms}.
If there is an $e \in E(v)$ on which $f$ is not a rotation, and an $e' \in E(v)$ on which $f$ is not a reflection, then $f$ does not map $R_j(v)$ to itself, for both $j \in \{0,1\}$.
\end{lemma}
\begin{proof}
Assume for a contradiction that $f$ maps each $R_j(v)$ to itself.
Since $f$ is not a rotation on $e$, there exist $a,b \in \bbZ_4$, $a \neq b$, and $s,t \in \bbZ_4$, $s < t$, such that $f(e,s) = (e,s+a)$ and $f(e,t) = (e,t+b)$. 
W.l.o.g.\ suppose that $0 \notin \{s,t\}$ (otherwise we can modify the first condition below so that $a_1 \notin \{s,t\}$, and reason analogously). Pick a tuple $\bar{\alpha} \coloneqq ((e,a_1),(e',a_2),...,(e_k,a_k)) \in R_0(v)$ with the following properties:
\begin{enumerate}
\item $a_1 = 0$.
\item $\sum \bar{\alpha} = \sum_{i=1}^k a_i = 0$.
\item $a_2$ is such that $f(e',a_2+1) = (e', f(a_2)+1)$.
\end{enumerate}
Such a tuple can always be chosen because $k \geq 3$, and the third property can be achieved because $f$ induces a permutation on $e'$ that is not a reflection, so it must satisfy case (1) of Lemma \ref{lem:permutationsOnZ4}.

Write $\bar{\alpha}^s$ and $\bar{\alpha}^t$ for the tuples obtained from $\bar{\alpha}$ by exchanging $(e,0)$ in the first entry of $\bar{\alpha}$ with $(e,s)$ or $(e,t)$, respectively. 
Then $\sum \bar{\alpha}^s = s$, and $\sum \bar{\alpha}^t = t$.
Define $\bar{\alpha}_{+1}$ as the tuple obtained from $\bar{\alpha}$ by exchanging the second entry $(e', a_2)$ with $(e', a_2+1)$. Thus, $\sum \bar{\alpha}_{+1} = 1$. Also, let $\bar{\alpha}^s_{+1}$ and $\bar{\alpha}^t_{+1}$ be the tuples obtained from $\bar{\alpha}_{+1}$ by exchanging the first entry with $(e,s)$ or $(e,t)$, respectively. So $\sum \bar{\alpha}^s_{+1} = s+1$, $\sum \bar{\alpha}^t_{+1} = t+1$.

Consider now the tuples $f(\bar{\alpha})$ and $f(\bar{\alpha}_{+1})$. Because $\bar{\alpha}, \bar{\alpha}_{+1} \in R_0(v)$ and we are assuming that $f$ preserves $R_0$, we have $f(\bar{\alpha}), f(\bar{\alpha}_{+1}) \in R_0(v)$. Moreover, by the third property of $\bar{\alpha}$, we know that $\sum f(\bar{\alpha}_{+1}) = \sum f(\bar{\alpha})+1$. Therefore, we must have $\sum f(\bar{\alpha}) = 0$ and $\sum f(\bar{\alpha}_{+1}) = 1$.  
The entry sums of $\bar{\alpha}, \bar{\alpha}^s, \bar{\alpha}^t$ are pairwise distinct (as we are assuming $0 \notin \{s,t\}$). That means that two of them are in the same relation $R_j(v)$ and one of them is in $R_{1-j}(v)$. 
Suppose first that $j=1$. Then, because $s<t$, we must have $s = 2$ and $t=3$. 
Then $\sum \bar{\alpha}^s_{+1} = 1+s = 3, \sum \bar{\alpha}^t_{+1} = 1+t = 0$. 
Now $\delta \coloneqq \sum f(\bar{\alpha}^t) - \sum f(\bar{\alpha}^s) = t-s+b-a = 1+b-a$, and we have $b-a \neq 0$. So $\delta \neq 1$. Because $f$ is assumed to preserve the relations, and $\bar{\alpha}^s, \bar{\alpha}^t \in R_1(v)$, also $f(\bar{\alpha}^s), f(\bar{\alpha}^t) \in R_1(v)$. Since $\delta \neq 1$, we can only have $\delta = 0$ or $\delta=3$.
The case $\delta = 0$ is ruled out as follows: Consider $\bar{\alpha}^s_{+1}$ and $\bar{\alpha}^t_{+1}$. 
Also for these tuples, we have $\sum f(\bar{\alpha}^t_{+1}) - \sum f(\bar{\alpha}^s_{+1}) = \delta$. So if $\delta = 0$, then $f(\bar{\alpha}^s_{+1})$ and $f(\bar{\alpha}^t_{+1})$ are in the same relation $R_j(v)$. But this cannot be the case because as already established, $\sum \bar{\alpha}^s_{+1} = 3$ and $\sum \bar{\alpha}^t_{+1} = 0$, so the preimages under $f$ are not in the same relation $R_j(v)$.
Therefore, we must have $\delta = 3$. This means that $\sum f(\bar{\alpha}^s) = 3$ and $\sum f(\bar{\alpha}^t) = 2$ as we have $\sum \bar{\alpha}^s = 2$ and $\sum \bar{\alpha}^t = 3$, and $f$ preserves $R_1(v)$. 
From this, we can conclude $\sum f(\bar{\alpha}^s_{+1}) = 0$ and $\sum f(\bar{\alpha}^t_{+1}) = 3$ using again the third property of $\bar{\alpha}$. But as already mentioned, we have $\sum \bar{\alpha}^s_{+1} = 3$, so $f$ maps $\bar{\alpha}^s_{+1}$ incorrectly. This finishes the case that $R_1(v)$ is the relation which contains two of the tuples $\bar{\alpha}, \bar{\alpha}^s, \bar{\alpha}^t$. If $R_0(v)$ contains two of these tuples, then $\sum \bar{\alpha}^s = 1$. In that case, a similar argument leads to a contradiction.
\end{proof}

\paragraph*{Properties of the constructed instances over $\bbZ_4$}
First, we need to prove that, as desired, $\StructA^\Mm (G)\not\cong \tilde{\StructA}^\Mm(G)$.  The following fact is standard (see~\cite[Corollary~2.13]{Diestel}

\begin{lemma}
\label{lem:existenceOfPerfectMatching}
Let $k\geq 3$. Every bipartite $k$-regular graph contains a perfect matching.
\end{lemma}

\begin{lemma}
\label{lem:TseitinSolutions}
Let $k \geq 3$ and let $G = (V,E)$ be a bipartite connected $\ell$-regular graph.
Let $\Ee$ be a system of linear equations over $\bbZ_4$ with variable set $X = \{x_e \mid e \in E(G)\}$ and one equation of the form $\sum_{e \in E(v)} x_e = c_v$ for every vertex $v \in V$, where $c_v \in \bbZ_4$ is a constant.
If there is a constant $a \in \bbZ_4$ such that $c_v = a$ for all but one vertex $\tilde{v} \in V$, where $c_{\tilde{v}} \neq a$, then $\Ee$ has no $\bbZ_4$-solution.
\end{lemma}
\begin{proof}
Let $S \uplus T = V$ be the bipartition.
Because $G$ is $k$-regular, $|T| = |E|/k = |S|$. 
Suppose for a contradiction that $\lambda \colon X \to \bbZ_4$ is a solution for $\Ee$. 
We assume w.l.o.g.\ that $\tilde{v} \in S$. Then we have $\sum_{v\in S} \sum_{e \in E(v)}\lambda(x_e) = (|S|-1) \cdot a + c_{\tilde{v}}$ and $\sum_{v\in T} \sum_{e \in E(v)}\lambda(x_e) = |T| \cdot a$. Since $|S| = |T|$ and $c_{\tilde{v}} \neq a$, these values are not equal. 
However, this is a contradiction since clearly $\sum_{v \in S} \sum_{e \in E(v)}\lambda(x_e)=\sum_{e\in E(G_n)}\lambda(x_e)=\sum_{v\in T} \sum_{e \in E(v)} \lambda(x_e)$.
\end{proof}

\begin{lemma}
\label{lem:tseitinAlmostSolution}
Let $k \geq 3$ and let $G$ be a $G$ be a bipartite connected $k$-regular graph.
Let $\Ee$ be a system of linear equations with variables $X$ over $\bbZ_4$ as in Lemma \ref{lem:TseitinSolutions}. Assume again that there exists an $a \in \bbZ_4$ and a $\tilde{v} \in V(G)$ so that $c_{\tilde{v}} \neq a$ and $c_v = a$ for every $v \in V(G) \setminus \{\tilde{v}\}$. 
Fix an arbitrary $v' \in V(G)$.
There is a $\lambda \colon X \to \bbZ_4$ that satisfies every equation in $\Ee$ except potentially the equation for $v'$, and $\sum_{e \in E(v')} \lambda(x_e) - c_{v'} \in \{ a-c_{\tilde{v}}, c_{\tilde{v}}-a\}$.
\end{lemma}
\begin{proof}
Suppose first that $a = 0$. Then consider the assignment $\lambda_0 \colon X \to \bbZ_4$ that is $0$ at every edge. It satisfies every equation except the one for $\tilde{v}$, where $\sum_{e \in E(v')} \lambda(x_e) - c_{\tilde{v}} = 0 - c_{\tilde{v}} = a -c_{\tilde{v}}$. This shows the lemma for $v' = \tilde{v}$. If $v' \neq \tilde{v}$, then by adding $c_{\tilde{v}}$ to an incident edge of $\tilde{v}$, and propagating this alternatingly with $-c_{\tilde{v}}$ along a path from $\tilde{v}$ to $v'$, we can define a new assignment $\lambda'$ that satisfies every equation except the one for $v'$. Moreover,  
$\sum_{e \in E(v')} \lambda(x_e) - c_{v'} \in \{ -c_{\tilde{v}}, c_{\tilde{v}}\}$ (depending whether the length of the path is odd or even).
In case that $a \neq 0$, we can take an assignment as in the case $a=0$, and add $a$ to $\lambda(x_e)$, for every edge $e$ of a fixed perfect matching of $G$ -- a perfect matching exists by Lemma \ref{lem:existenceOfPerfectMatching}.
\end{proof}

\begin{lemma}
\label{lem:nonIsomorphismMaltsevInstances}
Let $G=(V,E)$ be a connected $k$-regular bipartite graph.
Then for any $k \geq 3$, $\StructA^\Mm(G) \not\cong \tilde{\StructA}^\Mm(G)$.
\end{lemma}
\begin{proof}
Suppose for a contradiction that $f \colon A_G^\Mm \to A^\Mm_G$ is a bijection that defines an isomorphism from $\StructA^\Mm(G)$ to $\tilde{\StructA}^\Mm(G)$. Summation in the following is always in $\bbZ_4$, i.e.\ mod $4$.
Because of the preorder $\prec^{\StructA^\Mm(G)}$, $f$ must map the set $\{e\} \times \bbZ_4$ to itself, for each $e \in E$.
In the following, we also write $f(i)$ instead of $f(e,i)$ if the edge $e$ is clear from the context.\\
\\
\noindent
\textbf{Claim:} It is neither the case that $f$ is a rotation on every $e \in E$, nor that $f$ is a reflection on every $e \in E$.\\
\textit{Proof of claim.}
Suppose first that $f$ is a rotation on every edge. Let, for every $e \in E$, $a_e \in \bbZ_4$ be the cyclic shift applied to edge $e$, as described before.
By Lemma \ref{lem:localIsomorphisms}, the sum of the shifts must be $0$ at every $v \in V \setminus \{\tilde{v}\}$, and it must be $2$ at $\tilde{v}$.
For each $e \in E$, let $a_e \in \bbZ_4$ denote the value of the cyclic shift that $f$ applies to $\{e\} \times \bbZ_4$. Now the admissible values of the $a_e$ can be seen as a solution to a system of linear equations as in Lemma \ref{lem:TseitinSolutions}. But by that lemma, such a solution does not exist.
This shows that $f$ cannot be a rotation on every edge. 
If $f$ is a reflection on every edge, then according to Lemma \ref{lem:localIsomorphisms}, the sum of the shifts must be $1$ at every $v \in V \setminus \{\tilde{v}\}$, and it must be $3$ at $\tilde{v}$. 
So again, the shifts $a_e$, for each $e \in E(G)$, must be a solution to an equation system as in Lemma \ref{lem:TseitinSolutions}, which cannot exist.\\
\\
From the claim it follows that there is at least one edge $e$ on which $f$ is not a rotation. In fact, this $e$ can be chosen in such a way that one of its endpoints is a vertex $v \neq \tilde{v}$ such that not on all edges in $E \setminus \{e\}$, $f$ is a reflection. To see this, start with an arbitrary $e$ where $f$ is not a rotation. If it does not have an endpoint that satisfies the desired condition, then move from $e$ along a path consisting only of edges where $f$ is a reflection, until an edge $e$ is reached that does satisfy the condition. This is always possible because by the claim, there exist edges on which $f$ is not a reflection. 
Hence, we can find the following set-up: An edge $e \in E$ on which $f$ is not a rotation, and an endpoint $v \neq \tilde{v}$ of $e$ such that there is $e' \in E(v) \setminus \{e\}$ on which $f$ is not a reflection. 
Then by Lemma \ref{lem:notLocalIsomorphisms}, $f$ does not preserve $R_0(v)$ and $R_1(v)$. Since $v \neq \tilde{v}$, any isomorphism has to preserve these relations.
Therefore, $f$ cannot be an isomorphism between $\StructA^\Mm(G)$ and $\tilde{\StructA}^\Mm(G)$.
\end{proof}

Distinguishability of the structures in the logic $\Ll^k(\Qq_k)$ is now an immediate consequence:
\begin{corollary}
\label{cor:structuresDistinguishableInQk}
Let $G=(V,E)$ be connected, $k$-regular and bipartite.
Then $\StructA^\Mm(G)$ and $\tilde{\StructA}^\Mm(G)$ are distinguished by a sentence in $\Ll^k(Q_k)$.
\end{corollary}
\begin{proof}
Let $Q$ be a quantifier of arity $k$ corresponding to the class of structures isomorphic to $\StructA^\Mm(G)$ for some graph $G$.  Then, trivially, $Q$ is definable in $\Ll^k(\Qq_k)$.
\end{proof}

\paragraph*{Indistinguishability of the instances in the logic with partial Maltsev quantifiers}

To complete the proof of Theorem \ref{thm:MaltsevSeparation}, for every $k \geq 3$, we define for each $k$, a graph $G$ such that $\StructA^\Mm(G)$ and $\tilde{\StructA}^\Mm(G)$ cannot be distinguished in $\Ll^k(\Qq^\Mm_k)$. Here it suffices to use one base graph for each $k \in \bbN$, rather than a family of graphs, because the arity \emph{and} the number of variables are fixed to $k$. 

We call a $k$-regular graph $G$ \emph{near-$k$-connected} if the removal of any $k$ edges does not disconnect the graph, except in the case that the $k$ removed edges are the incident edges of a vertex $v$ (in which case $v$ becomes isolated from the rest).

For any $k \geq 3$ let $G^k$ be the graph obtained by removing from the biclique $K_{k+1,k+1}$ a set of edge $M$ that form a perfect matching.  Thus, $G^k$ is a $k$-regular bipartite graph on two sets (say $A$ and $B$) of $k+1$ vertices each.  

\begin{lemma}
  \label{lem:baseGraphsForMaltsev}
$G^k$ is near-$k$-connected. 
\end{lemma}
\begin{proof}
  Let $S$ be a set of edges of $G^k$, not all incident on a single vertex, with $|S| \leq k$.  We wish to argue that the graph $G'$ obtained by removing $S$ from $G^k$ is  connected.  Since $A$ has $k+1$ vertices, there is at least one $a \in A$ such that none of the edges of $S$ is incident on $A$.  We show that all vertices in $G'$ are reachable from $a$.  Let $b$ be the unique non-neighbour of $a$ in $B$ and let $A' = A\setminus \{a\}$ and $B' = B \setminus \{b\}$.  Then, every vertex in $B'$ is a neighbour of $a$ and hence reachable from $a$.  Moreover, if $u \in A'$ is such that some edge between $u$ and $B'$ is not in $S$, then $u$ is also reachable from $a$.  Hence, if for all $u \in A'$ there is at least one edge between $u$ and $B'$ that is not in $S$ then all of $A'$ is reachable from $a$.  Also, since all $k$ edges incident on $b$ are to $A'$ and they cannot all be in $S$, $b$ is also reachable from $a$ in $G'$.

  Thus, we are only left with the case that there is a vertex $u \in A'$ such that all $k-1$ edge between $u$ and $B'$ are in $S$.  Then, since $S$ cannot contain all edges incident on $u$, $b$ is still a neighbour of $u$ in $G'$.  Moreover, there cannot be more than one vertex in $A'$ satisfying these conditions.  Indeed, if $u,u' \in A'$ are such that each has $k-1$ incident edges in $S$, then $S$ has at least $2k-2$ edges (since the edges incident on $u$ and $u'$ form disjoint sets) and $2k-2 > k$ (since $k \geq 3$).  Now, since $S$ includes $k-1$ edges incident on $u$ and none of these is incident on $b$, $S$ can contain at most one edge incident on $b$.  Thus, $b$ still has at least two neighbours in $A'$, and the one that is not $u$ is reachable from $a$.  We conclude that $b$, and therefore also $u$, is reachable from $a$.
\end{proof}

\begin{lemma}
\label{lem:duplicatorWinningStrategyMaltsev}
For every $k \geq 3$, the Duplicator has a winning strategy in the game
$\Mm_k(\StructA^\Mm(G^k), \tilde{\StructA}^\Mm(G^k),\emptyset,\emptyset)$.
\end{lemma}
The winning strategy again consists in Duplicator maintaining a certain invariant. In the following, fix $k \geq 3$, and let $G = (V,E)$ be the graph $G^k$.

A bijection $f \colon A_G^\Mm \to A_G^\Mm$ is called \emph{edge-preserving} if $f$ maps the set $\{e\} \times \bbZ_4 $ to itself, for every $e \in E$. Duplicator only ever plays edge-preserving bijections. Again we have the notion of a bijection being good bar a vertex but it is slightly different than before:
In this proof, an edge-preserving bijection $f$ is called \emph{good bar $w$}, for a $w \in V$ if all of the following hold:
\begin{enumerate}
\item $f$ is either a rotation on every $\{e\} \times \bbZ_4$ or a reflection on every $\{e\} \times \bbZ_4$, for $e \in E$.
\item For every $v \in V \setminus \{w\}$, $\restrict{f}{A_v^\Mm} \colon \StructA^\Mm(v,s) \to \StructA^\Mm(v,t)$ is an isomorphism, where $s,t \in \bbZ_2$ are the charges of the gadgets in $\StructA^\Mm(G)$ and $\tilde{\StructA}^\Mm(G)$, respectively.
\item At $w$, the sum of the shift values of $f$ on the edges in $E(w)$, $c \coloneqq \sum_{e \in E(w)} c_e$,   
is precisely $2$ off from being a local isomorphism: That is, if $f$ is a rotation on every edge, then $c=0$ if $w= \tilde{v}$, and $c=2$ if $w \neq \tilde{v}$. If $f$ is a reflection on every edge, then $c=1$ if $w= \tilde{v}$, and $c=3$ if $w \neq \tilde{v}$ (cf.\ Lemma \ref{lem:localIsomorphisms}).
\end{enumerate}

After every round, let $\bar{\alpha}, \bar{\beta}$ denote the tuples of pebbles on $\StructA(G)$ and $\tilde{\StructA}(G)$, respectively. Let $F(\bar{\alpha}) \subseteq E$ denote the set of edges $e$ such that a pebble in $\bar{\alpha}$ lies on an element of $\{e\} \times \bbZ_4$.
Duplicator ensures the following \emph{invariant for the Maltsev game}.
There exists a $w \in V$ such that:
\begin{enumerate}
\item Every $e \in E(w)$ is not in $F(\bar{\alpha})$.
\item There is an edge-preserving bijection $f$ that is good bar $w$ and satisfies $f(\alpha) = \beta$.
\end{enumerate}
In the beginning of the game, this invariant is fulfilled for $w = \tilde{v}$ with $f \colon A_G^\Mm \to A_G^\Mm$ being the identity map.

\begin{lemma}
\label{lem:invariantMaintainedMaltsev}
If the invariant is satisfied in a game position $(\alpha, \beta)$ at the beginning of the round, then Duplicator has a strategy to ensure the invariant also at the beginning of the following round.
\end{lemma}
\begin{proof}
Spoiler chooses $r \in [k]$ and announces which $r$-tuple $\bar{y}$ of pebbles is moved this round, and whether he is playing a left or a right CSP quantifier move.
Duplicator plays the bijection $f$ that is good bar $w$ which exists by the invariant (it does not matter whether it maps from the left to the right structure or the other way round). 
Now Spoiler chooses an $r$-tuple in either $\StructA^\Mm(G)$ or $\tilde{\StructA}^\Mm(G)$, depending on which move he has chosen. Suppose w.l.o.g.\ that Spoiler is playing in $\StructA^\Mm(G)$, and let $\bar{a}$ be his chosen tuple of vertices. Let $\bar{\alpha}' \coloneqq \alpha[\bar{a}/\bar{y}]$ be the new tuple of pebbled vertices in $\StructA^\Mm(G)$.\\
\\
\textit{Case 1:} There is an edge $e \in E(w) \setminus F(\bar{\alpha}')$.\\
In this case, Duplicator plays $P \coloneqq \{f(\bar{a})\}$ i.e.\ effectively replies according to the bijection without using the power of the partial Maltsev move. Since $f$ defines a local isomorphism everywhere except at the tuples in $R_j(w)$, for $j \in \{0,1\}$, and because all edges incident to $w$ were pebble-free before Spoiler's move, the pebbles now still define a local isomorphism (since no tuple in $R_j(w)$ is fully pebbled).
Thus, Duplicator has not lost. It remains to define a new bijection $f'$ that is good bar $w'$ for some suitable $w' \in V$.
Since $f$ is good bar $w$, it is either a rotation or a reflection on every edge, and it does not map $R_j^{\StructA^\Mm(G)}(w)$ to $R_j^{\tilde{\StructA}^\Mm(G)}(w)$, for each $j \in \{0,1\}$. 
If $f$ is a rotation on every edge, then in order to change it to a local isomorphism from $R_j^{\StructA^\Mm(G)}(w)$ to $R_j^{\tilde{\StructA}^\Mm(G)}(w)$ for each $j \in \{0,1\}$, we have to make sure that the sum of shifts on the edges in $E(w)$ is $0$ if $w \neq \tilde{v}$, and $2$ if $w = \tilde{v}$.
If $f$ is a reflection on every edge, then the sum of shifts has to be $1$ if $w \neq \tilde{v}$, and $3$ if $w = \tilde{v}$ (both according to Lemma \ref{lem:localIsomorphisms}). In either case, this is achieved by letting $f'(e,i) \coloneqq (e,f(e,a)+b)$ for every $a \in \bbZ_4$ for a suitable choice of $b$ that corrects the sum of the shifts in $E(w)$ as desired. 
Let $v \neq w$ be the other endpoint of $e$. To compensate for the shift by $b$, we let $f'(e',a) \coloneqq (e',f(e',a)-b)$ for some other edge $e' \in E(v)$, and we proceed like this along a pebble-free escape path to some vertex $w'$ that satisfies condition 1 of the invariant. Such a path exists because $G$ is $k$-near-connected, and such a vertex $w'$ exists because there are only $k$ pebbles, but $k+1$ vertices on each side of the bipartition of $G$.

On all edges that are not on this escape path, we define $f'$ like $f$ (see also Lemmas \ref{lem:pathIso} and the proof of Lemma \ref{lem:duplicatorWins}).
By definition, the new bijection $f'$ is good bar $w'$, and $w'$ satisfies all conditions of the invariant for the new game position.\\
\\
\textit{Case 2:} Each $e \in E(w)$ is in $F(\bar{\alpha})$, that is, $E(w) = F(\bar{\alpha})$.\\
In this case, Spoiler can in principle expose the fact that $f$ is not a local isomorphism at $w$. Duplicator can prevent this by playing a Maltsev move. 
By the invariant, the sum of shifts $c$ that $f$ applies to the edges in $E(w)$ is off by $2$ from inducing a partial isomorphism. So by Lemma \ref{lem:localIsomorphisms}, $f$ maps $R^{\StructA(G)}_j(w)$ to $R^{\tilde{\StructA}(G)}_{1-j}(w)$ for each $j \in \{0,1\}$. Thus, $\bar{a} \in R_j^{\StructA(G)}(w)$ and $f(\bar{a}) \in R_{1-j}^{\tilde{\StructA}(G)}(w)$ for some $j \in \{0,1\}$. 
Let $\bar{a} = ((e_1,a_1),...,(e_k,a_k))$.
Assume first that $\theta \coloneqq \sum_{i \in [k]} f(a_i) \in \{1,3\}$. Then $\sum_{i \in [k]} a_i \in \{0,1\}$ in case that $\theta = 3$, and $\sum_{i \in [k]} a_i \in \{2,3\}$ in case that $\theta = 1$, because $f$ does not preserve $R_j^{\StructA(G)}(w)$.
That means, to turn $f$ into a local isomorphism, the value of $\theta$ must be increased by $1$ or $2$.

Let $f(\bar{a})_{+1,-}$ be a tuple like $f(\bar{a})$, but with $(e_1,f(a_1))$ replaced with $(e_1,f(a_1)+1)$.
Let $f(\bar{a})_{+1,+1}$ be like $f(\bar{a})$, but $(e_1,f(a_1))$ replaced with $(e_1,f(a_1)+1)$ and $(e_2,f(a_2))$ replaced with $(e_2,f(a_2)+1)$. Let $f(\bar{a})_{-,+1}$ be like $f(\bar{a})$ but $(e_2,f(a_2))$ replaced with $(e_2,f(a_2)+1)$.

Then Duplicator plays: $P = (f(\bar{a})_{+1,-}, f(\bar{a})_{+1,+1}, f(\bar{a})_{-,+1})$.
One can see that applying the Maltsev operation to these tuples yields $\bar{a}$, so this choice of $P$ is admissible. 
Now Spoiler picks a $\bar{b} = ((e_1,b_1), \dots, (e_k,b_k)) \in P$. In any case, $\sum_{i \in [k]} b_i \in \{1+\theta,2+\theta\}$. This means that $\bar{b} \in R_{j}^{\tilde{\StructA}(G)}(w)$, so the pebbles $\bar{a} \mapsto \bar{b}$ define a local isomorphism, as desired. 
It remains to show that Duplicator can choose a bijection $f'$ that satisfies the invariant in the next round. In case that $\bar{b} = f(\bar{a})_{+1,+1}$, we just define $f'(e_1,a) \coloneqq (e_1,f(e_1,a)+1)$ for all $a \in \bbZ_4$, and likewise on the edge $e_2$. Then we propagate these two shifts along two disjoint and pebble-free paths to an escape vertex $w'$ that satisfies condition 1 of the invariant for next round. This is possible because all the $k$ pebbles lie on $E(w)$ now, so all other edges are pebble-free, and we can find $2$ disjoint paths because $G$ is 2-connected.

Because $G$ is bipartite, the escape paths have the same length modulo 2 and so, indeed, the bijection $f'$ thus defined has shift value $2$ at $w'$ (cf.\ Lemma \ref{lem:pathIso}). At $w$, we have corrected the error of $2$ in the shift value, which existed by definition of $f$ being good bar $w$.
We furthermore define $f'$ like $f$ on all other edges that are not on any of the two escape paths. 

In case that $\bar{b} = f(\bar{a})_{+1,-}$ or $\bar{b} = f(\bar{a})_{-,+1}$, the situation is more difficult. Assume w.l.o.g.\ that $\bar{b} = f(\bar{a})_{+1,-}$, so the $+1$ has been added to the first entry of the tuple.
We have to define $f'$ so that it satisfies $\bar{b} = f'(\bar{a})$ and so that it defines a local isomorphism at $w$. This is not possible by adding cyclic shifts to $f$, but it requires to invert the current bijection $f$.
For each edge $e_i \in E(w)$, let $f_i \colon \bbZ_4 \to \bbZ_4$ denote the projection of $f$ to the $\bbZ_4$-part of $\{e_i\} \times \bbZ_4$.
We now define for each $e_i \in E(w)$ the bijection $f'_i \colon \bbZ_4 \to \bbZ_4$ that is the projection of the new $f'$ to the respective edge gadgets.
Recall that $\bar{a} = ((e_1,a_1),\dots,(e_k,a_k))$.
Let 
\[
f'_1(x) \coloneqq -f_1(x)+2f_1(a_1)+1.
\]
For every $1 < i \leq k$, let
\[
f'_i(x) \coloneqq -f_i(x)+2f_i(a_i).
\]
It is easy to see that, by definition, we have for every $i \geq 2$: $f'_i(a_i) = f(a_i)$, and $f'_1(a_1) = f(a_1)+1$. So if we define $f' \colon A_G^\Mm \to A_G^\Mm$ so that its restriction to the copy of $\bbZ_4$ in $\{e_i\} \times \bbZ_4$ is $f'_i$, for each $i \in [k]$, then, as desired, $f'(\bar{a}) = f(\bar{a})_{+1,-} = \bar{b}$.
Moreover, the $f'$ thus defined yields a local isomorphism that maps $R_j^{\StructA^\Mm(G)}(w)$ to $R_j^{\tilde{\StructA}^\Mm(G)}$, for each $j \in \{0,1\}$. To see this, let $c_e \in \bbZ_4$, for each $e \in E(w)$ denote the additive shift constant in the rotation or reflection induced by $f$ on $\{e\} \times \bbZ_4$. Likewise, let $c'_e \in \bbZ_4$ be that shift for $f'$.
Then:
\begin{align*}
\sum_{e \in E(w)} c'_e &= - \sum_{e \in E(w)} c_e+1+2\sum_{i \in [k]} f_i(a_i)\\
&= - \sum_{e \in E(w)} c_e+1+2\theta\\
&= -\sum_{e \in E(w)} c_e+3.
\end{align*}
The last equality is due to our assumption above that $\theta \in \{1,3\}$.
Now by property 3 of what it means for $f$ to be good bar $w$, the sum $\sum_{e \in E(w)} c_e$ is precisely $2$ off from yielding a local isomorphism. 
We go through one case as an example: Suppose $f$ is a rotation and $c= \sum_{e \in E(w)} c_e = 0$. Then any local isomorphism from $R_{j}^{\StructA}(G)(w)$ to $R_{j}^{\tilde{\StructA}}(G)(w)$ has shift value $2$ (according to Lemma \ref{lem:localIsomorphisms}, this is the case where $w = \tilde{v}$, and so, any local isomorphism must exchange $R_j(w)$ and $R_{1-j}(w)$). Then, $-\sum_{e \in E(w)} c_e+3 = 3$, which means by Lemma \ref{lem:localIsomorphisms} that $f'$ does $R_j(w)$ and $R_{1-j}(w)$, as desired. 
One can check that for all other values of $\sum_{e \in E(w)} c_e$, a similar effect takes place, and $f'$ always induces a local isomorphism. 
It only remains to extend the definition of $f'$ from $E(w)$ to all edges. Since there exist only $k$ pebbles, all of which lie on $E(w)$ now, we are free to choose $f'$ arbitrarily on every other edge. Indeed, we need to invert $f$ everywhere, so that if $f$ was a rotation on every edge, then $f'$ is a reflection on every edge, and vice versa. 
We simply choose $f'$ in such a way that its sum of shifts is correct at every vertex except at one arbitrarily chosen $w'$ that is not a neighbour of $w$. At $w'$, we ensure that the sum of shifts is off by $2$ from being a local isomorphism, as required by the invariant. 
This is possible by Lemma \ref{lem:tseitinAlmostSolution}: If $f'$ is a rotation on every edge, then its shift values of the edges can be written as a solution to a system of equations where every right hand side is $0$ except at $w'$, where it is $2$; if $f'$ is a reflection on every edge, the system of equations has $1$ in every right hand side except for $w'$, where it is $3$. 

This finishes the subcase where $\theta \in \{1,3\}$. The case where $\theta \in \{0,2\}$ is symmetric. Then, Duplicator chooses $P = (f(\bar{a})_{+3,-}, f(\bar{a})_{+3,+3}, f(\bar{a})_{-,+3})$, and the rest of the above proof goes through similarly, where one has to exchange the role of $1$ and $3$ everywhere.
\end{proof}

Now the \textbf{proof of Theorem \ref{thm:MaltsevSeparation}} is a combination of the preceding results: For every $k \in \bbN$, choose $G$ according to Lemma \ref{lem:baseGraphsForMaltsev}. Then $\StructA^\Mm(G)$ and $\tilde{\StructA}^\Mm(G)$ are not isomorphic (Lemma \ref{lem:nonIsomorphismMaltsevInstances}) but indistinguishable by any sentence in $\Ll^k(\Qq^\Mm_k)$ (Lemma \ref{lem:duplicatorWinningStrategyMaltsev}). Finally, by Corollary \ref{cor:structuresDistinguishableInQk}, there is a sentence in the logic $\Ll^k(\Qq_k)$ that does distinguish $\StructA^\Mm(G)$ from $\tilde{\StructA}^\Mm(G)$.

\end{document}